\documentclass[a4paper]{llncs}

\usepackage{a4wide}
\usepackage{times}
 
\usepackage{xspace}
\usepackage[english]{babel}
\usepackage[latin1]{inputenc}

\usepackage[svgnames]{xcolor}
\usepackage{subfigure}
\usepackage{amsmath,amssymb}
\usepackage{enumerate}
\usepackage{url}

\usepackage{tikz} 
\usetikzlibrary{shapes.misc}
\usetikzlibrary{calc}
\usetikzlibrary{automata}
\usetikzlibrary{backgrounds}
\usetikzlibrary{decorations.pathreplacing}
\usetikzlibrary{fadings}
\usetikzlibrary{automata,positioning,shapes.geometric,fit}

\tikzset{initial text={},
    every state/.style={circle,minimum size=.4cm,draw=blue!50,very thick,fill=blue!20},
    secret/.style={minimum size=.4cm,draw=red!50,very thick,fill=red!20,rectangle},
    node distance=1.5cm,on grid,auto,
    bend angle=65}

\def\malar{M\"alardalen University\xspace}
\def\ie{{i.e.},~}
\def\eg{{e.g.},~}
\def\st{{s.t.}~}

\def\unk{\bot} 
\def\start{i_1} 

\def\et{\mbox{\normalfont\textsf{time}}}
\def\stree{\mbox{\normalfont\textsf{tree}}}
\def\wcet{\normalfont\text{WCET}}

\def\flag{\textit{flags}}
\def\setNZ{\textit{SetStatusB}}
\def\cmpU{\textit{cmpU}}
\def\NDcmp{\textit{NDcmp}}
\def\update{\textit{update}}

\def\tr{\textit{TR}}
\def\runs{\textit{Runs}}

\def\lang{\calL}

\def\proj{\textit{proj}} \def\run{\textit{run}}

\newcommand{\setB}{\mathbb B}

\def\calI{{\cal I}}

\def\calD{{\cal D}}
\def\calM{{\cal M}}
\def\calC{{\cal C}}

\def\calP{{\cal P}}
\def\calR{{\cal R}}

\def\calL{{\cal L}}
\def\calV{{\cal V}}
\def\calS{{\cal S}}

\def\li{\calL\calI}
\def\endef{\ifmmode\squareforged\else{\unskip\nobreak\hfil
\penalty50\hskip1em\null\nobreak\hfil$\blacksquare$
\parfillskip=0pt\finalhyphendemerits=0\endgraf}\fi}

\def\true{\mbox{\textsc{true}\xspace}}
\def\false{\mbox{\textsc{false}\xspace}}

\usepackage{listings}

\definecolor{gris}{rgb}{0.3, 0.3, 0.3}

\lstdefinelanguage{AssemblerARM9}{%
basicstyle={\fontsize{5}{6}\selectfont\sffamily},
commentstyle={\color{gris}\it},
morekeywords={push,pop,word,ldr,str,ldreq,add,sub,b,bl,bx,bxeq,bne,mov,cmp,movgt,ble,cmple,stmdb,ldmdb,ldm,mvn,ldrep,subeq,beq,subgt,addle,ldmia,andeq,bgt},
keywordstyle={\textbf},%
sensitive=false,%
backgroundcolor=\color{blue!10!white},
rulecolor=\color{blue!50!white},
fillcolor=\color{blue!20!white},
texcl=true,
xleftmargin=0.1cm,
flexiblecolumns=true,
morecomment=[s][\it\color{gris}]{/*}{*/},
showstringspaces=true,
escapechar=\%
}

\renewcommand*\thelstnumber{${\the\value{lstnumber}}\!\!:$}
\newcommand{\lline}[1]{\the\value{#1}}

\makeatletter
\newcommand{\linelabel}[1]{%
\addtocounter{figure}{1}%
\immediate\write\@auxout{\string\newlabel{#1}{{\the\value{lstnumber}}{\thepage}%
{Line numbering}{figure.\thefigure}{}}}%
\addtocounter{figure}{-1}%
}
\makeatother

\newcommand{\FC}[1]{\textcolor{red}{#1}}

\newcommand{\sem}[1]{[\![#1]\!]}  
\newcommand{\isem}[1]{<\!\!#1\!\!>}

\newcommand{\plusminus}{\raisebox{.7mm}{\tiny $+/-$}}

\def\uppaal{\textsc{Uppaal}\xspace}

\def\emptyset{\varnothing}

\newcommand{\Paragraph}[1]{\vskip.2em\noindent{\bfseries \em #1.}}
   
\title{Computation of WCET using Program Slicing and Real-Time Model-Checking}

\author{Jean-Luc Béchennec, Franck Cassez\thanks{Author supported by a Marie Curie
    International Outgoing Fellowship within the 7th
    European Community Framework Programme.}}
\institute{{L'UNAM University, CNRS, IRCCyN,
 Nantes, France\\
}
}

\newcommand{\fname}[1]{#1} 
\begin{document}

\maketitle

\thispagestyle{empty}

\begin{abstract} 
  Computing accurate WCET on modern complex architectures is a
  challenging task.  A lot of attention has been devoted to this
  problem in the last decade but there are still some open issues.
  First, the control flow graph (CFG) of a binary program is needed to
  compute the WCET and this CFG is built using internal knowledge of
  the compiler that generated the binary code; moreover once
  constructed the CFG has to be manually annotated with loop bounds.
  Second, the algorithms to compute the WCET (combining Abstract
  Interpretation and Integer Linear Programming) are tailored for
  specific architectures: changing the architecture (\eg replacing an
  ARM7 by an ARM9) requires the design of a new ad hoc algorithm.
  Third, the tightness of the computed results (obtained using the
  available tools) are seldom compared to actual execution times
  measured on the real hardware.

  In this paper we address these problems. We first describe a fully
  automatic method to compute a CFG based solely on the binary program
  to analyse.  Second, we describe the model of the hardware as a
  product of timed automata, and this model is independent from the
  program description.  The model of a program running on a hardware
  is obtained by synchronizing (the automaton of) the program with the
  (timed automata) model of the hardware.  Computing the WCET is
  reduced to a reachability problem on the synchronised model and
  solved using the model-checker \uppaal.  Finally, we present a
  rigorous methodology that enables us to compare our computed results
  to actual execution times measured on a real platform, the
  ARM920T. 
    %
\end{abstract}


\noindent\textcolor{green!80!black}{Updated version, \today}

\section{Introduction}\label{sec-intro}

Embedded real-time systems are composed of a set of tasks (software)
that run on a given architecture (hardware).  These systems are
subject to strict timing constraints that must be enforced by a
scheduler.
Determining if a given scheduler can schedule the system is possible only if some
bounds are known about the execution times of each task.  Performance
wise, determining tight bounds is crucial as using rough
over-estimates might either result in a set of tasks being wrongly
declared non schedulable,
or leads to the choice of an overpowered and expensive hardware where
a lot of computation time is lost.


\Paragraph{\bfseries The WCET Problem}
Given a program $P$, some input data $d$ and the hardware $H$, the
\emph{execution-time} of $P$ on input $d$ on $H$,
is measured as the number of cycles of the fastest
component of the hardware \ie the processor.  
The program is given in binary code or equivalently in the assembly
language of the target processor.\footnote{When we refer to the
  ``source'' code, we assume the program $p$ was generated by a
  compiler, and refer to the high-level program (\eg in C) that was
  compiled into $P$.}  The \emph{worst-case execution-time} of program
$P$ on hardware $H$, $\wcet(P,H)$, is the supremum on all input data
$d$, of the execution-times of $P$ on input $d$ for $H$.
%
%
The WCET problem asks the following: Given $P$ and $H$, compute
$\wcet(P,H)$.

In general, the WCET problem is undecidable because otherwise we could
solve the halting problem.  However, for programs that always
terminate and have a bounded number of paths, it is 
computable.  Indeed the possible runs of the program can be
represented by a finite tree.
Notice that this does not mean that the problem is tractable though.

If the input data are known or the program execution time is
independent from the input data, the tree contains a single path and
it is usually feasible to compute the WCET.  Likewise, if we can
determine some input data that produces the WCET (which can be as
difficult as computing the WCET itself), we can compute the WCET on a
single-path program.

It is not often the case that the input data are known or that we can
determine an input that produces the WCET.  Rather the (values of the)
input data are unknown, and the number of paths to be explored might
be extremely large: for instance, for a Bubble Sort program with $100$
data to be sorted, the tree representing all the runs of the
(assembly) program on all the possible input data has more than
$2^{50}$ nodes.  Although symbolic methods (\eg using BDDs) can be
applied to analyse some programs with a huge number of states, they
will fail to compute the exact WCET on Bubble Sort by exploring all
the possible paths.

Another difficulty of the WCET problem stems from the increasingly
complex architectures embedded real-time systems are running on.  They
feature multi-stage \emph{pipelines} and fast memory components like
\emph{caches} that both influence the WCET in a complicated manner.
It is then a challenging problem to determine a precise WCET even for
relatively small programs running on complex architectures.

\Paragraph{\bfseries Methods and Tools for the WCET Problem}
The reader is referred to~\cite{wcet-survey-2008} for an exhaustive
presentation of WCET computation techniques and tools.  There are two
main classes of methods for computing WCET:
\begin{itemize}
\item Testing-based methods. These methods are based on experiments
  \ie running the program on some data, using a simulator of the
  hardware or the real platform. The execution time of an experiment
  is measured and, on a large set of experiments, maximal and minimal
  bounds can be obtained. A maximal bound computed in this way is
  \emph{unsafe} as not all the possible paths have been explored.
  These methods might not be suitable for safety critical embedded
  systems but they are versatile and rather easy to implement.

  RapiTime~\cite{rapitime} (based on pWCET~\cite{pWCET}) and
  Mtime~\cite{mtime} are measurement tools that implement this
  technique.

\item Verification-based methods.  These methods often rely on the
  computation of an \emph{abstract} graph, the control flow graph
  (CFG), and an abstract model of the hardware.  Together with a
  static analysis tool they can be combined to compute WCET.  The CFG
  should produce a super-set of the set of all feasible paths. Thus
  the largest execution time on the abstract program is an upper bound
  of the WCET.  Such methods produce \emph{safe} WCET, but are
  difficult to implement. Moreover, the abstract program can be
  extremely large and beyond the scope of any analysis. In this case,
  a solution is to take an even more abstract program which results in
  drifting further away from the exact WCET.

  Although difficult to implement, there are quite a lot of tools
  implementing this scheme: Bound-T~\cite{bound-T},
  OTAWA~\cite{otawa-2010}, TuBound~\cite{tubound},
  Chronos~\cite{chronos}, 
  SWEET~\cite{sweet-2003} and
  aiT~\cite{aiT,wcet-ai-aswsd-ferdinand-04} are static analysis-based
  tools for computing WCET.
\end{itemize}

The verification-based tools mentioned above rely on the construction
of a control flow graph, and the determination of loop bounds.  This
can be achieved using user annotations (in the source code) or
sometimes inferred automatically.  The CFG is also annotated with some
timing information about the cache misses/hits and pipeline stalls,
and paths analysis is carried out on this model \eg by Integer Linear
Programming (ILP).  The algorithms implemented in the tools use both
the program and the hardware specification to compute the CFG fed to
the ILP solver.  The architecture of the tools themselves is thus
monolithic: it is not easy to adapt an algorithm for a new processor.
This is witnessed by the \emph{WCET'08 Challenge
  Report}~\cite{wcet-chal-report-08} that highlights the difficulties
encountered by the participants to adapt their tools for the new
hardware in a reasonable amount of time.  Moreover, the results of the
computation are not  compared to actual execution times measured
on a real platform.  Notice that aiT reports comparisons with
ARMulator (the ARM simulator of the RealView Development Suite) but
this simulator is not cycle accurate as emphasised in ARM documentation
for ARMulator, Application Notes~93~\cite{arm-app-note-93}:
\begin{quote}\small \em
  ARMulator consists of C based models of ARM cores and as such cannot
  be guaranteed to completely reproduce the behaviour of the real
  hardware. If 100\% accuracy is required, an HDL model should be
  used.
\end{quote}

\Paragraph{\bfseries Outline of the Paper}
Section~\ref{sec-related-work} presents our contribution and related
work.  In Section~\ref{sec-archi} we give the specification of the
hardware we use in the experiments. Section~\ref{sec-semantics} gives
some formal definitions for program execution on a given hardware.
Section~\ref{sec-wcet} presents a modular way to compute the WCET of a
given program.  Section~\ref{sec-cfg} presents the technique we use to
automatically build the CFG.  Section~\ref{sec-hw} gives the \uppaal
timed automata models of the hardware.  In Section~\ref{sec-implem} we
report on the implementation and tool chain we have
developed. Section~\ref{sec-experiments} describes the methodology we
use to compare our computed WCET with actual WCET and contains a
summary of the results.  Section~\ref{sec-conclu} concludes with our
ongoing and future work..

\section{Related Work}\label{sec-related-work}
\Paragraph{\bfseries WCET and Model-Checking}
Only a few tools use model-check\-ing techniques to compute WCET.
Considering that ($i$) modern architectures are composed of
\emph{concurrent} components (the units of the different stages of the
pipeline, the caches) and ($ii$) the \emph{synchronization} of these
components depends on \emph{timing constraints} (time to execute in
one stage of the pipeline, time to fetch data from the cache), formal
models like \emph{timed automata}~\cite{AD94} and state-of-the-art
\emph{real-time model-checkers} like
UPPAAL~\cite{uppaal-sttt-97,uppaal-40-qest-behrmann-06} appear
well-suited to address the WCET problem.

In~\cite{wcet-cav-metzner-04}, A.~Metzner showed that model-checkers
could well be used to compute safe WCET on the CFG for programs
running on pipelined processors with an instruction cache. More
recently, Lv \emph{et al.}~\cite{lv-rtss-2010} combined AI techniques
with real-time model-checking (and \uppaal) to compute WCET on
multicore platforms.

In~\cite{huber-wcet-09}, B.~Huber and M.~Schoeberl consider Java
programs and compare ILP-based techniques with model-checking
techniques using the model-checker UPPAAL.  Model-checking techniques
seem slower but easily amenable to changes (in the hardware model).
The recommendation is to use ILP tools for large programs and
model-checking tools for code fragments.


\medskip

The use of timed automata (TA) and the model-checker UPPAAL for
computing WCET on pipelined processors with caches was reported
in~\cite{metamoc-2009,metamoc-2010} where the METAMOC method is
described. METAMOC consists in: 1) computing the CFG of a program, 2)
composing this CFG with a (network of timed automata) model of the
processor and the caches.  Computing the WCET is then reduced to
computing the longest path (timewise) in the network of TA.

The previous framework is very elegant yet has some shortcomings: (1)
METAMOC relies on a value analysis phase that may not terminate, (2)
some programs cannot be analysed (if they contain register-indirect
jumps), (3) some manual annotations are still required on the binary
program, \eg loop bounds and (4) the \emph{unrolling} of loops is not
safe for some cache replacement policies (FIFO).

In a previous work~\cite{cassez-acsd-11} we have already reported some
similar results on the computation of WCET using TA.
In~\cite{cassez-acsd-11}, what is similar to METAMOC is the use of
network of timed automata to model the cache\footnote{Note that a
  similar model is reportedly due to A.~P.~Ravn
  in~\cite{huber-wcet-09}.} and pipeline stages.  However, in this
preliminary work we had chosen to: (1) build the CFG without any need
for annotations and (2) use a new and very compact encoding of the
program and pipeline stages' states. In contrast  METAMOC uses a values
analysis phase and requires loop bounds annotations to obtain an
(unfolded) graph of the program.

\Paragraph{\bfseries Our Contribution}
Compared to our previous work~\cite{cassez-acsd-11}, this paper
contains three new original contributions: (1) an automatic method to
compute a CFG and a reduced abstract program equivalent WCET-wise to
the original program; (2) detailed hardware formal models and (3) a
rigourous methodology to make it possible the comparison of computed WCET
to actual WCET measured on a real hardware.


\section{Architecture of the ARM920T}\label{sec-archi}
The development board we model and use in the experiment section. It
is an Armadeus APF9328 board~\cite{armadeus} which bears a 200MHz
Freescale MC9328MXL micro-controller with an ARM920T processor. The
processor embeds an ARM9TDMI core that implements the ARM v4T
architecture.  An overview of the ARM920T architecture is given in
Fig.~\ref{fig:arm920T}.  The component we model in
Section~\ref{sec-hw-model} are highlighted in orange.
\begin{figure}[htbp] 
   \centering
   \includegraphics[width=.9\textwidth]{\fname{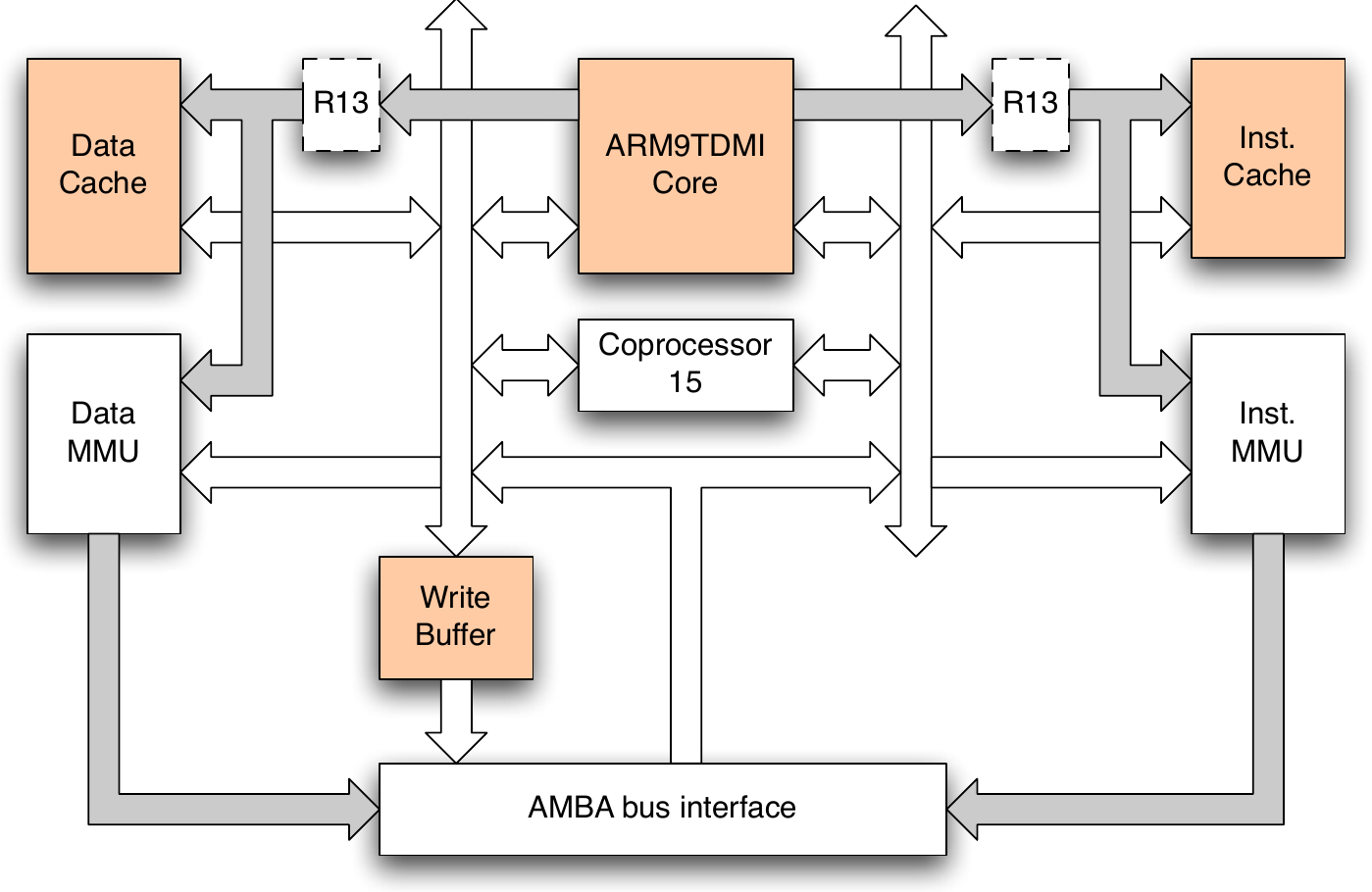}} 
   \caption{Simplified block diagram of the ARM920T. Gray arrows are
     address buses/connections. White arrows are data
     buses/connections. Both are 32 bits wide. Coprocessor 15 hosts
     control registers for the caches and the MMUs. Actually Register
     R13 (which should not be confused with the ARM9TDMI {\em r13} register
     presented in section \ref{sec:isa})
     is not duplicated and is located in Coprocessor 15. It hosts
     a process ID used for virtual address to physical address
     translation. Some blocks like the Write back Physical TAG RAM and
     various debug and/or coprocessor interfaces are not shown.}
   \label{fig:arm920T}
\end{figure}

\subsection{Reduced Instruction Set Computer Architecture}
\label{sec:isa}

The ARM architecture is a {\em Reduced Instruction Set Computer}
(RISC) architecture. The instruction set consists of fixed size
instructions and a few simple addressing modes.  There are $16$
general purpose registers $r_0$ to $r_{15}$, specialized memory
transfer instructions (load/store), and data-processing instructions
that operate on registers only.
Other interesting features are \emph{multiple} load/store instructions
and \emph{conditional} execution of instructions (to improve data and
execution throughput).

Three of the general purpose registers are used in a specialized way.
Register $r_{13}$ is the {\em stack pointer} (we use {\em sp} in the
sequel to refer to this register).  Register $r_{14}$ is the {\em link
  register} ({\em lr} in the sequel) and hosts the return address of
function calls.  Register $r_{15}$ is the {\em program counter} ({\em
  pc} in the sequel).  

An instruction is defined by a mnemonic\footnote{And the condition and
  flags (like the ``s'' flag).} (\eg \emph{mov}) and the operands.  In
the sequel, we let $\calR=\{r_0,\cdots,r_{12},$ $sp,lr,pc\}$ be the set
of registers of the architecture and $\calI$ be the (finite) set of
RISC instructions.

\subsection{Execution Pipeline}
\label{sec:pipeline}

The ARM920T uses a 5-stage \emph{execution pipeline}, the purpose of
which is to execute concurrently the different tasks (Fetch, Decode,
Execute, Memory, Writeback) needed to perform an instruction. The
(normal) flow of instructions in the pipeline is shown in
Fig.~\ref{fig:arm920Tpipeline}. 
\begin{figure}[htbp] 
   \centering
   \includegraphics[scale=.7]{\fname{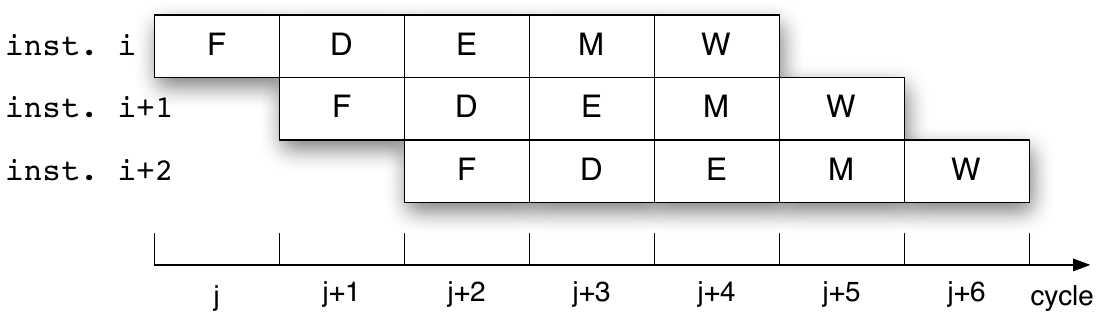}} 
   \caption{Pipeline of the ARM920T: Instruction is fetched in
     F. Instruction decode and operand register accesses are done in
     D. Execution is done in E. Load/store instructions do their
     memory accesses in M. Results are written back to registers in
     W.}
   \label{fig:arm920Tpipeline}
\end{figure}
%
This optimal flow may be slowed down when pipeline \emph{stalls}
occur.
Most of the time, two independent consecutive instructions do not
incur a \emph{stall} and the throughput is 1
instruction/cycle. However in certain cases stalls can occur.
\begin{figure}[htbp] 
   \centering
   \includegraphics[scale=.7]{\fname{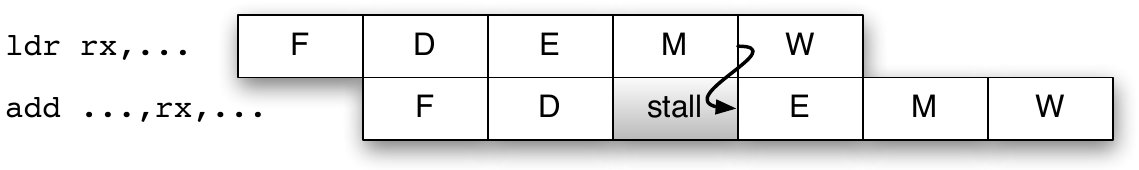}} 
   \caption{Load delay pipeline hazard.}
   \label{fig:pipeLoadDelay}
\end{figure}
Assume instruction \texttt{ldr r1,[sp,\#0]} is followed by \texttt{add
  r0,r1,\#1}: the first instruction loads register $r_1$ (with the
content of a memory cell) and the second uses $r_1$ to compute
\emph{r0}.  This sequence of instructions brings about a {\em load
  delay} depicted on Fig.~\ref{fig:pipeLoadDelay}.  One \emph{stall}
cycle is inserted before processing instruction \texttt{add r0,r1,\#1}
because the load instruction produces the operand needed ($r_1$) at
the E stage of the add
instruction at the end of its M stage.
%

Sometimes the target address of a branch instruction is produced at
the end of the E stage (\eg conditional branching that needs the
result of a comparison operation). The ARM920T does not implement any
\emph{branch prediction} mechanism. As a consequence fetching the next
instruction can only be done after the branch instruction has
completed the E stage: this causes a {\em branch delay} depicted
Fig.~\ref{fig:pipeBranchDelay} that results in 2 stall cycles before
the fetch of the branch target instruction can be performed.

\begin{figure}[htbp] 
   \centering
   \includegraphics[scale=.7]{\fname{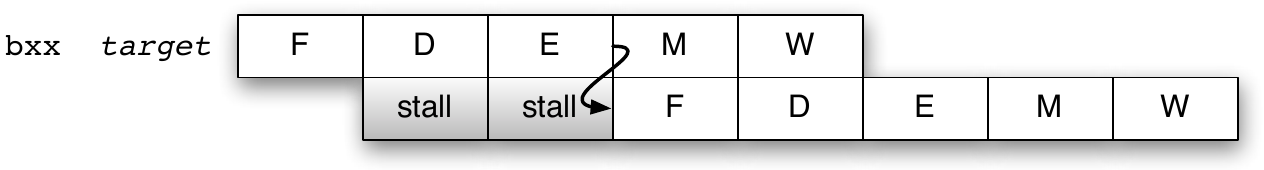}} 
   \caption{Branch delay pipeline hazard.}
   \label{fig:pipeBranchDelay}
\end{figure}

\subsection{Main Memory, Instruction and Data Cache, Write Buffer}
\label{sec-cache}
Both instruction and data caches have the same architecture. They are
16KB, 8-ways set associative caches. There are 64 sets and 512 32
bytes long lines.  Replacement policy may be set to pseudo-random or
round-robin (FIFO).  Both caches implements
\emph{allocate-on-read-miss} \ie a data is inserted in the cache if
missing when a read is performed.

The data cache may be configured in \emph{write-through} (when data in
the cache is modified, it is immediately written to the main memory)
or \emph{write-back} (modified cached data are only written to main
memory when needed) but does not implement
\emph{allocate-on-write-miss}: if non cached data is written to, they
are not cached but instead written to main memory directly.  So, even
if configured in write-back, a write miss acts as a write-through.
Each data cache line has 2 \emph{dirty} bits (indicating that a cached
item has been modified since last cached), one per half-line, to
indicate the half-line must be written back when it is replaced.

A 16-word \emph{write buffer} helps to reduce stalls when a write to
the main memory occurs because of a write miss or, if the cache is
configured in write-back, when a dirty line has to be replaced. The
write buffer is organized in 4 half-line entries to allow cache
write-back on a half-line basis.

Finally, transfers between the caches and main memory are serialized
and the bus abstracted away.



\section{Program Semantics}\label{sec-semantics}
In this section we present the formal semantics for the execution of
binary programs. 
We make the following assumptions on the binary programs we
analyse:
\begin{itemize}
\item[(A1)] the termination of a program does not depend on input data, \ie
  a program terminates for all input data; and
\item[(A2)] reference to stack values is via the specialised register
  \emph{sp} only.\footnote{Note that these assumptions are not
    compulsory but they are made in the current implementation of our
    tool in the Compute CFG component (See Section~\ref{sec-implem}).
    Moreover, they are satisfied by programs obtained using a compiler
    conforming to the ARM ABI~\cite{armabi}.  However, the technique
    described here can be extended to encompass a more general
    framework.}
\item[(A3)] references to memory cells are independant from input
  data. This ensures that when an instruction computes the address of
  a memory cell, it is always defined.
\item[(A4)] The programs do not contain recursive calls.
\end{itemize}

\subsection{Notations}
We let $\setB=\{\true,\false\}$.  We have already introduced some
notations: $\calR$ is the set of registers of the hardware, $\calI$ is
the (finite) set of instructions the hardware can perform, and $\calM$
the (finite) set of main memory cells the program can access.  In the
sequel we will introduce a set of predicates $\calP$ and for $x \in
\calR \cup \calP \cup \calM$ (set of registers or predicates or memory
cells), $\sem{x}$ denotes the content of $x$.  A \emph{program state}
$s$ is a valuation of the variables in $\calR \cup \calP \cup \calM$
\ie a mapping from $\calR \cup \calP \cup \calM$ to $\calD$ where
$\calD$ is a finite set \eg 32-bit integers.  We let $\calS$ be the
set of program states.

As program instructions are located in main memory, we define the set
of labelled instructions $\li=\calM \times \calI$ to be the set of
pairs $(\ell:i)$ indicating that instruction $i$ is stored at address
$\ell$ in main memory.  Consequently, a program $P$ is simply a subset
(necessarily finite) of $\li$.  We use the notation $\isem{\iota}:
\calS \rightarrow \calS$ to denote the \emph{semantics} of instruction
$\iota \in \li$.

\begin{remark}
  The semantics is defined on labelled instructions which means that
  the semantics itself may depend on the address the instruction is
  stored at.  This is actually the case for some instructions like
  \texttt{12: ldr r0,[pc,\#4]} the semantics of which is ``load
  register $r_0$ with the content of the memory cell located at offset
  $4$ from the current value of \emph{pc}'', \ie at address
  $12$.\footnote{In pipelined architecture, the actual memory address
    is translated due to pipelining. For example in the ARM9, the
    address is \emph{$8+$ the offset} that appears in the
    instruction.}
\end{remark}

\subsection{Example Program}
As a running example we take the binary program FIBO of
Listing~\ref{lis-fibo-30}. It has been compiled (\texttt{gcc}) and
de-assembled (\texttt{objdump}) using the GNU ARM tools from
Codesourcery~\cite{codesourcery}.  It computes $\textit{Fib}(30)$ the
Fibonacci number $u_{30}$ with $u_0=1$, $u_1=1$ and
$u_n=u_{n-1}+u_{n-2}, n \geq 2$.
A program is stored in memory and the memory address of each
instruction is the leftmost decimal number.\footnote{Instructions
  addresses are multiple of $4$ in the ARM 32-bit instruction set.}
Each program has a designated initial instruction
$\iota_0=(\ell_0,i_0)$ and $\sem{\emph{pc}}=\ell_0$ at the beginning
of the execution of the program.

To give the semantics of programs, we assume there is a set of
variables $\calP=\{le,gt,\cdots\}$ to hold the truth values of the
predicates used in the conditional instructions of the
program\footnote{In the ARM 32-bit instruction set, the truth values
  of these predicates are stored in the status bits N, Z, C, V.}.

The semantics of program FIBO is given in terms of assignments to
registers (on the right-hand side of each instruction in
Listing~\ref{lis-fibo-30}). Each instruction assigns a new value to
register \emph{pc}: except for branching instructions the assignment
is $\sem{pc}:=\sem{pc}+4$ and we omit it in this case.  A comparison
operator (\eg line~24) sets the truth value of the predicates that are
used later in the program (\eg \emph{eq} for instruction at line~24).
The main loop of the computation is between line~24 and~52; notice
that this optimized program (compiled with option -O2) computes
$u_{n+2}$ in each round of this loop ($r_2$ holds the value of $n$ and
is incremented twice in the body of the loop).

\begin{center}
\begin{minipage}[t]{0.6\linewidth}
\begin{lstlisting}[language=AssemblerARM9,numbers=none,caption={Program FIBO: computes Fibonacci
$30$},label={lis-fibo-30},basicstyle={\fontsize{8}{9}\selectfont\ttfamily}]
0 <main>: /* starts at address 0;  %$\sem{lr}$% is the return address */
0   mov    r1,#30       %$\sem{r_1} := 30$ %
4   mov    r2,#2        %$\sem{r_2} := 2$ %
8   add    r1,r1,#1     %$\sem{r_1} := \sem{r_1} + 1$ %      
12  add    r2,r2,#1     %$\sem{r_2} := \sem{r_2} + 1$ %
16  mov    r0,#0        %$\sem{r_0} := 0$ %       
20  mov    r3,#1        %$\sem{r_3} := 1$ % 
24  cmp    r2,r1        %$\sem{eq} := (\sem{r_2}=\sem{r1})$ %  
28  add    r0,r0,r3     %$\sem{r_0} := \sem{r_0} + \sem{r3}$ % /* %$r_0=u_n$% */      
32  bxeq   lr           %if $(\sem{eq})$ $\sem{pc}:=\sem{lr}$ else $\sem{pc}:=36$ %
36  add    r2,r2,#1    
40  add    r2,r2,#1            
44  add    r3,r3,r0     %$\sem{r_3} := \sem{r_3} + \sem{r_0}$ %
48  cmp    r2,r1        %$\sem{eq} := (\sem{r_2}=\sem{r1})$ %        
52  add    r0,r0,r3    
56  bne    24           %if $(\neg\sem{eq})$ $\sem{pc}:=24$ else $\sem{pc}:=60$ %        
60  bx     lr           %$\sem{pc}:=\sem{lr}$%    
\end{lstlisting}
\end{minipage}
\end{center}

\subsection{Abstract Hardware Model}
The real hardware (Section~\ref{sec-archi}) consists of the pipelined
processor, instruction and data caches, write buffer and the main
memory of the computer.
We abstract away the details of the communication medium (AMBA bus,
MMU\footnote{The MMU is considered to be programmed to make a
  translation from a virtual address page $v$ to the physical address
  page $p$ such as $p = v$.}).  We choose to treat the content of the
main memory as a component of the \emph{program state} and thus it is
not part of the state of the hardware.  The same remark applies for
the register and we consider they are part of the program state.  A
\emph{state} of the hardware is then defined by the states of the
different stages of the pipeline and the states of the caches.

As we are only interested in computing execution times, we can
consider that the hardware is an abstract machine $H$ that reads
sequences of \emph{triples} $(\iota,A,d) \in \li \times \calM \times
\setB$ and outputs the \emph{time} it takes to process such a
sequence.  A triple $(\iota,A,d)$ consists of a (labelled) instruction
$\iota=(\ell:i)$ with $\ell \in \calM$ and $i \in\calI$ that
references a set of (main memory) addresses in $A$ and is performed if
$d=\true$; if $d=\false$ the instruction is a conditional instruction
and the condition the instruction depends on last evaluated to
$\false$.
Such triples (and sequences thereof) contain enough information to
compute the execution time:
\begin{itemize}
\item pipeline stalls (see~\ref{sec:pipeline}) can be inferred
  from the first component of the triple $\iota$ that contains the full
  text of the instruction and thus the read/written registers and the
  value of $d$ (whether the instruction is executed or not);
\item cache hits/misses (see~\ref{sec-cache}) are completely
  determined by the set $A$.
\end{itemize}
Examples of instructions for program FIBO (Listing~\ref{lis-fibo-30})
are \texttt{0: mov r1,\#30} and \texttt{32:bxeq lr}.  Notice that
there is no need for actual register values in $H$ neither for
performing the real computation as the timing of instructions in the
pipeline and the cache is fully determined by the instruction (and its
location in memory), whether it is performed or not (there are
conditional instructions), the registers read from/written
to\footnote{Some instructions (MUL/MLA/SMULL) have data dependent
  durations. In this case an upper bound can be used or a
  non-deterministically chosen value (see Section~\ref{sec-hw-model}
  for details).} and the memory addresses used in the instruction. The
fact that there is no branch prediction in the pipeline of the
hardware in the ARM920T makes things simpler but the framework we
present extends to the case with branch prediction
(see~\cite{cassez-acsd-11}).

The execution time of a sequence of triples also depends on the
initial state $\gamma$ of the hardware $H$.  Given a finite sequence
$w=q_0q_1q_2 \cdots q_n \in (\li \times \calM \times \setB)^*$ and an
initial state $\gamma$ of $H$, $\et_H(\gamma,w)$ is the execution time
of $w$ from initial state $\gamma$ of $H$.  It can be defined
precisely using for instance the HDL model of the hardware.  Notice
that at this point, we do not require sequences of triples to be
actual sequences produced by program $P$.

\subsection{Trace Semantics of a Program}

The execution of a program can be defined by an alternating sequence
of \emph{program states} and instructions.  
%
A \emph{run} of a program $P$ is a sequence $ \varrho = s_0 \; \iota_0
\; s_1 \; \iota_1 \; s_2 \; \iota_2 \quad \cdots \quad s_{n-1} \;
\iota_{n-1} \; s_n$ where $s_k$ is a program state and
$\iota_k=(\ell_k:i_k)$ is a labelled instruction with $\ell_k=s_k(pc)$
and such that $s_{k+1}=\isem{\iota_k}(s_k)$.
%
We let $\runs(P)$ be the set of runs of $P$.

The \emph{trace}, $\tr(\varrho)$, of the run $\varrho$ is the sequence
$q_0q_1q_2 \cdots q_n \in (\li \times \calM \times \setB)^*$ with
$q_i=(\iota_i,A_i,d_i)$ where $A_i$ is the set of memory addresses
referenced by instruction $\iota_i$ in state $s_i$ and $d_i \in \setB$
indicates whether the instruction is actually executed\footnote{$S_i$
  and $d_i$ can always be computed from $s_i$ and $\iota_i$.}.
For instance, instruction \texttt{ldr r0,[sp, \#4]}\footnote{The
  semantics is $s(r_0) := s ( s(sp)+4 )$.} from a program state $s$
with $s(sp)=12$ references address $16$ and is performed
(unconditional).
Instruction \texttt{128: addle r1,r1,\#1} is performed only if the
last comparison set the predicate \emph{le} $\true$ and from program
state $s$ the next triple in the trace is ($\texttt{128: addle
  r1,r1,\#1}$, $\emptyset$, $s(le))$.  As there are \emph{multiple}
load and store instructions, we need \emph{sets} of addresses to
represent the memory cells referenced by an instruction: instruction
\texttt{stm sp,\{r0,r1\}}\footnote{The semantics is $s(s(sp)):=s(r_1)$
  and $s((s(sp)-4):=s(r_0)$.} references addresses $12$ and $8$.

The execution time of a run $\varrho$ of $P$ from initial state
$\gamma$ of $H$ is defined by $\et_H(\gamma,\tr(\varrho))$.

Program $P$ has a set of \emph{initial} states $I$ (where
$\sem{\emph{pc}}$ gives the initial instruction of $P$) and the
contents of the registers, predicates and main memory can be in a
finite set of values. Notice that there can be many initial states as
the input data of $P$ can range over large sets. $P$ also has a set of
\emph{final} states, $F$, and we assume it can be defined using the
value of register \emph{pc} which gives the last instruction of $P$.
The language $\lang_I^F(P)$ of $P$ is the set of traces generated by
runs of $P$ that start in $I$ and end in $F$ \ie $\lang_I^F(P)= \{
\tr(\varrho)\ | \ \varrho= s_1 \iota_1 \cdots s_{n-1}\iota_{n-1}s_n,
\varrho \in \runs(P), s_1 \in I, s_n \in F\}$.  As we assume that $P$
always terminates for any input data, this language is finite (because
the set of memory contents is finite).

\section{Computation of the WCET}\label{sec-wcet}
\subsection{Modular Definition of WCET}
Given a run $\varrho$ of $P$, the execution time of $\varrho$ on $H$
from state $\gamma$ only depends on $\tr(\varrho)$.  This implies that
the WCET of $P$ only depends on $\lang_I^F(P)$ and the initial state
$\gamma$ of $H$.  Consequently if $\lang_I^F(P)$ is finite
\begin{equation} \label{eq-wcet-mod}
\wcet(P,H)=\max_{w \in \lang_I^F(P)} \et_H(\gamma,w)\mathpunct.
\end{equation}
The computation of $\wcet(P,H)$ thus amounts to ($i$) generating
$\lang_I^F(P)$, ($ii$) feeding $H$ with each $w \in \lang_I^F(P)$ and
tracking the maximal execution time.
This gives a \emph{modular} way of computing $\wcet(P,H)$ since a
generator for $\lang_I^F(P)$ and the behaviour of the abstract
hardware $H$ to be fed with $\lang_I^F(P)$ can be given independently
of each other.

\subsection{Extended Domain Abstraction}
In order to take into account all the possible values of the input
data, we use an \emph{extended domain} for the values of the main
memory cells.  We assume here that the values of the registers and
predicates are known in the initial state.

Let $\calD_\bot=\calD\cup \{\bot\}$ be the extended domain with $\bot$
the \emph{unknown value}. The semantics of instructions is extended to
this extended domain: for instance, the semantics of \texttt{add
  r0,r1,\#1} is given by
\[\sem{r_0}= \textit{ $\bot$ if $(\sem{r_1}=\bot)$
  and $\sem{r_1}+1$ otherwise}\mathpunct.
\]  
The semantics of comparison
instructions \eg \texttt{cmp r0,r1} is extended as well to
$\calD_\bot$ \eg for instruction 24 of program FIBO, 
\[
\sem{eq}= \bot
\textit{ if $((\sem{r_0}=\bot)$ or $(\sem{r_1}=\bot))$ and ($\sem{r_0}=\sem{r_1}$)
  otherwise}\mathpunct.
\]
When a conditional instruction is encountered and the condition is
$\bot$, the extended semantics of the instruction considers two
successors: one where the condition is $\true$ and the other where the
condition is $\false$.  If a branching instruction like $\mathtt{bx
 \  lr}$ is encountered and $\sem{lr}=\bot$ the next instruction is
undefined (we can encode this by jumping to a special ``error'' state
but this situation will not occur in the sequel).

We may now define an extended \emph{symbolic} semantics for a program
$P$, and starting from an initial state $s_0 : \calR \cup \calP \cup
\calM \rightarrow \calD_\bot$, the symbolic semantics define a
\emph{set} of runs (non-determinism may arise if some conditions are
tested and unknown).

Assume that the values of the registers and predicates are fixed in
the initial program and given by $s_0(\calR \cup \calP)$ and the input
data is $d$: the initial state of the memory is $s_0(d)$ with $s_0(d):
\calM \rightarrow D$.  The initial state of the program is thus
defined by $s_0(\calR \cup \calP)\cdot s_0(d)$.

Define $s_0^\bot : \calR \cup \calP \cup \calM \rightarrow D_\bot$ by:
$s_0^\bot(x)=s_0(x)$ for $x \in \calR \cup \calP$ and
$s_0^\bot(y)=\bot$ for $y \in \calM$.

The important property of the extended semantics is $(\pi)$: if $\varrho$ is a
run of $P$ from state $s_0(d)$, then $\varrho$ is a run of $P$ from
$s_0^\bot$ in the extended symbolic semantics.

In the sequel we write $\lang_\bot(P)$ for $\lang_{\{s_0^\bot\}}^F(P)$
and $\wcet_\bot(P,H)=\max_{w \in \lang_{\bot}(P)}
\et_H(\gamma,w)$. The property $(\pi)$ of the symbolic semantics
implies that $\lang_I^F(P) \subseteq \lang_{\bot}(P)$ and by language
inclusion we have
\begin{equation} \label{eq-wcet-mod-extended-dom} \wcet(P,H)=\max_{w
    \in \lang_I^F(P)} \et_H(\gamma,w) \leq \max_{w \in
    \lang_{\bot}(P)} \et_H(\gamma,w) = \wcet_\bot(P,H) \mathpunct.
\end{equation}
We can thus reduce the computation of (an upper bound of the)
$\wcet(P,H)$ to a symbolic simulation of program $P$ on the extended
domain $\calD_\bot$ from a unique initial state $s_0^\bot$.

As we have assumed that termination does not depend on the input data,
but is guaranteed for each program $P$, the symbolic simulation of $P$
on the extended domain terminates as well.  Each test that ensures
termination in $P$ cannot evaluate to $\bot$ because otherwise it
would depend on the input data and this would contradict assumption
(A1).

\subsection{WCET Computation as a Reachability Problem}
\label{sec-reach}
We can reduce the computation of the WCET to a \emph{reachability}
problem on a network a timed automata.  Indeed, as $\lang_{\bot}(P)$
is finite, it can be generated by a finite automaton
$\textit{Aut}(P)$.  The hardware $H$ (including pipeline, caches and
main memory) can be specified by a network of \emph{timed automata}
$\textit{Aut}(H)$ (formal models are given in Section~\ref{sec-hw}).
Feeding $H$ with $\lang_{\bot}(P)$ amounts to building the
\emph{synchronised} product $\textit{Aut}(H) \times \textit{Aut}(P)$.
On this product we define \emph{final states} to be the states where
the last instruction of $P$ flows out of the last stage of
pipeline. Assume a fresh\footnote{$x$ is not a clock of
  $\textit{Aut}(H)$.} clock $x$ is reset in the initial state of
$\textit{Aut}(H) \times \textit{Aut}(P)$. The WCET of $P$ on $H$ is
then the largest value, $\max(x)$, that $x$ can take in a final state
of $\textit{Aut}(H) \times \textit{Aut}(P)$ (we assume that time does
not progress from a final state).

We can compute $\max(x)$ using model-checking techniques with the tool
\uppaal~\cite{uppaal-40-qest-behrmann-06} (see
Section~\ref{sec-implem}).  To do this, we check a reachability
property ``(R): Can we reach a final state with $x \geq K$?'' on
$\textit{Aut}(H) \times \textit{Aut}(P)$.  If the property is true for
$K$ and false for $K+1$, $K$ is the WCET of $P$.  We can compute this
maximal value using the \emph{sup} operator that gives the maximal
value a clock can have in a reachable state.

Notice that to do this we have to explore the whole state
space\footnote{Checking that (R) is false or computing \emph{sup
    clock} implies the exploration of all the reachable states.} of
$\textit{Aut}(H) \times \textit{Aut}(P)$.  This means that to handle
large case studies, we need to reduce the state space as much as
possible.

An important point to notice is that the tightness of the WCET we
compute depends on an accurate description of $H$. The more precise
(time-wise) $\textit{Aut}(H)$ is, the more precise the computed WCET
will be.  It is thus not reasonable to take a very abstract $H$ (\eg
with caches that always miss) as it will give poor WCET estimates.
We can still have some control on the automaton $\textit{Aut}(P)$ that
generates the traces to be fed to $\textit{Aut}(H)$.  Indeed, we
should avoid generating two runs with the same trace as it will give
the same WCET (from the same initial state of $H$). This means that
\emph{minimizing} $\textit{Aut}(P)$ can effectively reduce the state
space (at least the number of paths explored in the product
$\textit{Aut}(H) \times \textit{Aut}(P)$).
In the next section we describe how to compute a reduced program $P'$
that generates the same set of traces as $P$.

\section{Slicing} \label{slicing} \emph{Program Slicing} was
introduced by Mark Weiser~\cite{weiser-84} in 1984.  The purpose of
program slicing is to compute a program \emph{slice} (by removing some
statements of the original program) \st the slice computes the same
values for some variables at some given statements. Program slicing is
often used for checking properties of programs. The reader is refered
to~\cite{Tip95asurvey} 
for a survey on the principles of (static and dynamic) slicing.

\subsection{Overview of Program Slicing}
In this section, we assume that we have the control flow graph of $P$,
$CFG(P)$, which is a directed graph, the nodes of which are in $P$.
$CFG(P)$ has a single \emph{entry} node (initial instruction of the
program $P$) and a single \emph{exit} node (that indicates the end of
program $P$). An example of a CFG for the Fibonacci program of
Listing~\ref{lis-fibo-30} is given in Figure~\ref{fig-cfg-fibo}.

\begin{figure}[thbtp]
  \centering
  \includegraphics[scale=0.45]{\fname{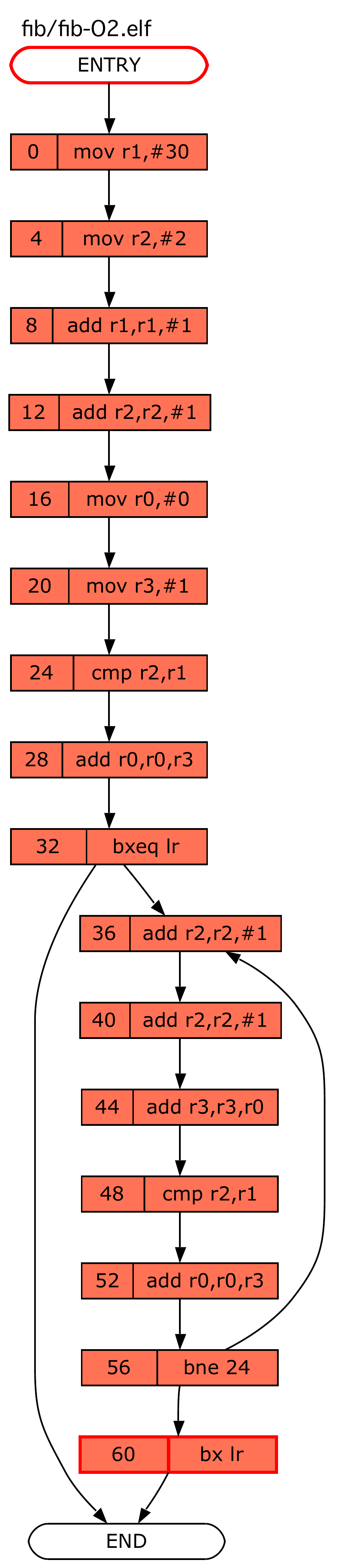}}
  \caption{CFP for Program of Listing~\ref{lis-fibo-30}.}
  \label{fig-cfg-fibo}
\end{figure}

A \emph{slice criterion} $\calC$ for $P$ is a subset $I' \subseteq P$,
and for each instruction $\iota \in I'$ an associated subset of
``variables'' $\calV(\iota) \subseteq \calR \cup \calP \cup \calM$. We
assume that $\calV(\iota)$ is actually included in the set of
registers that the instruction operates on but this is inessential.
For instance, a slice criterion for program $FIBO$ of
Figure~\ref{fig-cfg-fibo} can be instruction $48: \mathtt{cmp\ r2,
  r1}$ and associated set $\{r_1,r_2\}$.

Given input data $d \in \calD$, we write $\run(P,d)$ to denote the
(unique) run of $P$ on $d$.  Let $S \subseteq P$.  The
runs\footnote{Notice that at this stage, $\run(S,d)$ may not be finite
  and program $S$ may not terminate.} of $P$ and $S$ on input data $d
\in \calD$ are denoted
\begin{eqnarray*}
  \run(P,d) & = & s_1 \iota_1 s_2 \iota_2 \cdots s_k \iota_k \cdots s_{n} \iota_{n} s_{n+1} \\
  \run(S,d) & = & s'_1 \iota'_1 s'_2 \iota'_2 \cdots s'_k \iota'_k \cdots s'_{m} \iota'_{m} s'_{m+1} 
\end{eqnarray*}
Define the \emph{projection} $\proj(s,\iota)$ for a pair $(s,\iota)
\in \li$ by:
\[
\proj(s,\iota) =
\begin{cases}
  \varepsilon \text{ if $\iota \not\in S$ } \\
  (\proj_{\calV(\iota)}(s),\iota) \text{ otherwise}.
\end{cases}
\] 
\ie instructions not in the subset $S$ are ignored (replaced by
$\varepsilon$, the empty word) and for instructions in $S$ we keep the
projection on $\calV(\iota)$ of the program state. $\proj$ is extended
in the natural way to traces and we let
$\proj^*(\varepsilon)=\varepsilon$ and
$\proj^*(w.(s,\iota))=\proj^*(w).\proj^*(s,\iota)$ with $s$ a program
state and $\iota \in \li$.

$S$ is a \emph{slice} of $P$ for the slice criterion $\calC$ if it
satisfies, for every input data $d \in \calD$:
\begin{enumerate}
\item if $P$ terminates on input $d$ then $S$ terminates on input
  $d$ and
\item $\proj^*(run(P,d))=\proj^*(run(S,d))$. Notice that by definition
  of $\proj$, all the instructions of $S$ are in $\proj^*(run(S,d))$
  but the projection restricts the set $s'_k$ to the variables in
  $\calV(\iota'_k)$.
\end{enumerate}
In the sequel we recall how to (effectively) compute a slice for $P$
given a slice criterion $\calC$.

\subsection{Prerequisites for Computing a Program Slice}
The computation of a slice is based on an iterative solution of
\emph{dataflow} equations on the set of relevant variables for each
instruction in the CFG of $P$.  The relevant variables for an
instruction are the variables read from/written to by the instruction.
Due to the particular nature of binary programs, the knowledge of
relevant variables for an instruction might not be explicit: consider
the instruction $\textit{foo}=\texttt{str r0,[sp,\#4]}$ again.  This
instruction reads register $r_0$ and writes to the ``variable'' which
is the memory cell at location $\sem{\textit{sp}}+4$. This value is
not known at compile time. The previous instruction writes in the
\emph{stack} which is particular region of the main memory.  Other
instructions like \texttt{16: str r2,[r1, r3 lsl \#2]} might (read or)
write to arbitrary memory cells: in this case the memory cell with
address\footnote{The operator $<<$ denotes the logical shift left.}
$\sem{r_1} + (\sem{r_3} << 2)$.

In our approach we make the following choice: 
\begin{enumerate}
\item we consider that the content of the main memory outside the
  stack is always $\bot$; this means that we do not need to store the
  main memory content into the program state as it is constant.
\item by assumption (A2), every access to a stack value is via
  register \emph{sp}. We use the term \emph{stack reference} for
  instructions that read/write \emph{sp} and main memory reference for
  the other memory accesses.
\item for an instruction which has a stack reference, (\eg in
  \texttt{\texttt{str r0,[sp,\#4]}}), we only know the actual offset
  at runtime.  To define the referenced variables, we introduce a
  variable \emph{stack}. This means that we track the stack content in
  the state of the program and this variable is updated by instruction
  that do stack references.  The previous instruction thus reads $r_0$
  and $sp$ and writes to \emph{stack}.
\end{enumerate}

This enables us to define formally the set of \emph{referenced} and
\emph{defined} variables for each instruction, which is mandatory in
order to compute automatically a slice.

Given instruction $i \in \calI$ the set of read from (REF) and written
to (DEF) variables is given by:
\begin{itemize}
\item for instructions that do not make main memory references or
  stack references, \eg $i=\texttt{add r2,r1,\#1}$ we have
  $REF(i)=\{r_1\}$ and $DEF(i)=\{r_2\}$.
\item for instructions that make stack references, \eg
  $i=\mathtt{push(r_0,r_1)}$, we define $REF(i)=\{r_0,r_1,sp\}$ and
  $DEF(i)=\{sp,stack\}$.
\item for instructions that make main memory references, we assume the
  content of main memory is $\bot$. For an instruction like
  $i=\texttt{str r2,[r1, r3 lsl \#2]}$ we thus have
  $REF(i)=\{r_1,r_2,r_3\}$ and $DEF(i)=\emptyset$. Indeed, even if the
  memory location $\sem{r_1}+(\sem{r_3} << 2)$ is written to, the new
  content of the main memory does not depend on the values of the
  registers and thus we can omit it in the set of $DEF$ variables.
\end{itemize}

\subsection{Step~1: A Slice for Values of Register \emph{sp}.}
\label{sec-step1}
The first task we perform on a binary program $P$ is to compute the
possible values of the stack references (values of \emph{sp}).
 

We can compute the possible values of the stack pointer \emph{sp} for
a given instruction using a slice criterion $\calC$: $\calC$ contains
all the instructions that read/write the variable \emph{sp} \ie all
the instructions \st $DEF(i) \cap \{sp\} \neq \emptyset$ or $REF(i)
\cap \{sp\} \neq \emptyset$ .
 
We compute a slice of $P$ for $\calC$, $S_\calC(P)$ using the standard
definition of \emph{data dependence} and \emph{control dependence}
(see~\cite{Tip95asurvey}).  Once computed, we do a symbolic simulation
of $S_\calC(P)$ and track the values of \emph{sp} encountered for each
instruction in $\calC$.

As we have assumed (A1) that termination does not depend on the input
data, but is guaranteed for each program $P$, the symbolic simulation
of the slice $S_\calC(P)$ on the extended domain terminates as well.
During the course of the symbolic simulation, we track the values of
the \emph{sp} register for each stack reference instruction. At the
end of the simulation, we obtain the set of possible values for
\emph{sp} at each stack reference instruction.  Because of the
property of the slice, and the symbolic simulation in the extended
domain (superset of the set of runs) we can ensure that the set of
\emph{sp} values we obtain for each instruction is a superset of the
set actual values in $P$.

\paragraph*{Limitations.}
The previous approach works correctly if the stack is referenced only
via the register \emph{sp}, assumption (A2). This is ensured by the
API of the compilers from C/C++ to ARM for instance and thus is a
perfectly reasonable assumption.

\smallskip

We can take advantage of the computation performed previsouly to
narrow the $DEF$ and $REF$ variables for each instruction in $P$.
Assume for instruction $\iota=\texttt{4: str r0,[sp,\#4]}$, the set of
possible values of \emph{sp} is $\{12,16\}$. What we know about the
written to variables is more precise than being \emph{somewhere in the
  stack}. We know that variables at index $12$ and $16$ may be written
to, and this instruction does not modify other stack items at other
offsets.  We thus refine the definitions of $REF$ and $DEF$ for
instruction $\iota$ by setting: $REF^*(\iota)=\{r_0,sp\}$ (unchanged
in this case) and $DEF^*(\iota)=\{stack_{12},stack_{16}\}$.  This more
precise definitions will result in smaller subsequent slices as they
will introduce less data dependences in the CFG of a program.


In the sequel, we show how to use
program slicing to compute a WCET-equivalent program.  In the next
section, we also show how to iteratively use program slicing to build
the CFG of arbitrary assembly (unstructured) programs.

\subsection{Step~2: Using Program Slicing to Compute a WCET-Equivalent Program
}
\label{sec-wcet-equiv}
As in the previous subsection, assume that we have the complete CFG of
$P$ (building this CFG is addressed in Section~\ref{sec-cfg}).
Equation~\ref{eq-wcet-mod} implies that for any two programs $P$ and
$P'$,
\begin{equation}\label{eq-3}
  \lang_\bot(P) = \lang_\bot(P') \Longrightarrow \wcet_\bot(P,H) =
  \wcet_\bot(P',H)\mathpunct. 
\end{equation}
%

What we would like to do is to compute such a WCET-equivalent program
$P'$ which (hopefully) operates on a reduced subset of the set of
registers $\calR$ yet contains enough information to generate
$\lang_{\bot}(P)$.

Using the previously computed attributes $REF^*$ and $DEF^*$, we can
compute a WCET-equivalent program using an ad-hoc slice criterion
$\calC'$: $\calC'$ contains ($i$) all the instructions that perform
main memory transactions (including the stack), and each instruction
has the associated set of variables that defines the memory location,
($ii$) all the conditional instructions $\iota$ with associated set of
variables $\calV(\iota) \ni p$ if $p$ is the condition of the
instruction.\footnote{For a conditional memory transaction
  instruction, both the registers that are needed to compute the
  referenced memory address(es) and the condition are in the
  associated variables.}  For instance, instruction $j=(\texttt{16:
  ldr r2,[r1, r3 lsl \#2]})$ is in $\calC'$ and we have to track the
values of registers $\calV(j)=\{r_1,r_3\}$ since the memory address is
defined by $r_1$ and $r_3$.  For an instruction like
$l=(\texttt{12:addle r1,r2,\#1})$ we set $\calV(l)=\{le\}$.

Let $S_{\calC'}(P)$ be the slice computed using the criterion
$\calC'$.  What we want is to generate the language $\lang_\bot(P)$
using the slice.  For each instruction $\iota \in P$ we define a
corresponding abstracted  $\alpha(\iota)$ as follows:
\begin{itemize}
\item if $\iota \in P \cap S_{\calC'}(P)$ then $\alpha(\iota)=\iota$;
\item for the other instructions $\iota \in P \setminus
  S_{\calC'}(P)$, $\alpha(\iota)=\iota_{nop}$ where $\iota_{nop}$
  denote the instruction with the exact same syntax as $\iota$ but the
  semantics of $\iota_{nop}$ is $\sem{pc} :=\sem{pc} + 4$.  As the
  syntax of $\iota$ is identical to $\iota_{nop}$, this alos preserves
  the $REF$ and $DEF$ attributes.
\end{itemize}
We let $\alpha(P)$ be the program that comprises of instructions
$\alpha(\iota), \iota \in P$.  Notice that $alpha$ is one-to-one
mapping and thus we can consider $\alpha^{-1}$ when needed.

We can now prove the following Lemmas:
\begin{lemma}\label{lem-1}
  Let $\varrho=s_0^\bot \iota_0 s_1 \iota_1 s_2 \iota_2 \ \cdots \
  \iota_{k-1} s_k$ be a run of $P$. The run $\varrho'={s'}_0^\bot
  \alpha(\iota_0) s'_1 \alpha(\iota_1) s'_2 \alpha(\iota_2)$ $ \cdots
  \ \alpha(\iota_{k-1}) s'_k $ is in $\alpha(P)$ and
  $\tr(\varrho)=\tr(\varrho')$.
\end{lemma}
\begin{proof}
  We prove the Lemma by induction. The induction hypothesis (IH) is:
  for runs of length $k$, $\tr(\varrho)=\tr(\varrho')$ and
  $s'_k=\proj_{\calV(\iota_{k})}(s_k)$ if $\iota_k$ is the instruction
  following $\iota_{k-1}$ and $\iota_k$ is in the slice, and
  $s'_k=\proj_{\{pc\}}(s_k)$ otherwise.  The Lemma is true for runs of
  length $0$.  Assume we have a run of length $k+1$ \ie
  $\varrho=s_0^\bot \iota_0 s_1 \iota_1 s_2 \iota_2 \ \cdots \
  \iota_{k-1} s_k \iota_k s_{k+1}$.  First notice that instruction
  $\alpha(\iota_k)$ is a successor of $\alpha(\iota_{k-1})$ as the CFG
  of $\alpha(P)$ is isomorphic to $CFG(P)$.  We can compute the triple
  $(x,A,d)$ and $(x',A',d')$ added to the trace of $\varrho$ and
  $\varrho'$ after instructions $\iota_k$ and $\alpha(\iota_k)$:
  \begin{itemize}
    \item the first component of the triple is the same as $alpha(\iota_k)$ and $\iota_k$
      have exactly the same syntax (and location); hence $x=x'$.
    \item for the second component, memory references, there are two cases:
      \begin{itemize}
      \item either $\iota_k$ does not make any memory transfer and references only registers.
        The same applies to $\alpha(\iota_k)$ and the second component is the empty set;
      \item or $\iota$ has memory references. In this case, the
        registers that generate the memory references are in the
        slice, and thus the values at $s_k$ and $s'_k$ coincide by the
        ``projection'' property of the slice.  
      \end{itemize}
      In each case $A=A'$.
    \item the third components $d,d'$ are the values of the conditions
      of the instruction $\iota_k$ and $\alpha(\iota_k)$. If the
      instruction $\iota_k$ is unconditional, $d=\true$ and
      $d'=\true$.  Otherwise, the two instructions have the condition
      $c$. As $\iota_k$ is conditional, the condition is in the slice
      (by definition of the slice) and thus $s_k(c)=s'_k(c)$.  Hence
      $d=d'$. 
  \end{itemize}
  This proves that $\tr(\varrho)=\tr(\varrho')$ and completes the
  proof.  \qed
\end{proof}

\begin{lemma}\label{lem-2}
  Let $\varrho'={s'}_0^\bot \iota'_0 s'_1 \iota'_1 s'_2 \iota'_2$ $
  \cdots \iota_{k-1} s'_k$ be a run of $\alpha(P)$.  There is a run
  $\varrho=s_0^\bot \iota_0 s_1 \iota_1 s_2 \iota_2 \cdots \iota_{k-1}
  s_k$ of $P$ with $\iota'_i=\alpha(\iota_i)$ and
  $\tr(\varrho)=\tr(\varrho')$.
\end{lemma}
\begin{proof}
  The proof relies on the following fact: every instruction in $P$
  that has more than one successor is conditional and thus is in the
  slice.\footnote{We omit here the case of \emph{switch} statements
    but they are processed in a similar way and this is implemented in
    our tool.}
  Consequently, given two instructions $\iota'_j=\alpha(\iota_j)$ and
  $\iota'_{j+1}=\alpha(\iota_{j+1})$ in the slice, there is a unique
  sequence (with no loop) of instructions in $CFG(P)$ between
  $\iota_j$ and $\iota_{j+1}$.  This shows that there is a (unique)
  run in $P$ defined by $\iota_j=\alpha^{-1}(\iota'_j)$.
  Using the result of Lemma~\ref{lem-1} completes the proof. 
\qed
\end{proof}

By combining Lemma~\ref{lem-1} and~\ref{lem-2} we obtain:
\begin{theorem}
  $\wcet_\bot(P,H)=\wcet_\bot(\alpha(P),H)$.
\end{theorem}
\begin{proof}
  Lemmas~\ref{lem-1} and~\ref{lem-2} imply that
  $\lang_\bot(P)=\lang_\bot(\alpha(P))$ and by Equation~\ref{eq-3},
  the result follows. \qed
\end{proof}

When we do program slicing, many operations on registers are avoided
if they do not influence the control flow. The result is that
$\alpha(P)$ generates less states than $P$: assume register $r_4$ is
never used in $\alpha(P)$ but used in $P$, then all the states of $P$
that differs only on $r_4$ are collapsed into the same state.  This
also means that the automaton $Aut(\alpha(P))$ that generates
$\lang_\bot(\alpha(P))$ will have less states than $Aut(P)$.  Quite
often, some registers are not used at all or do not influence the
control flow and this reduces drastically the number of states in
$Aut(\alpha(P))$.

An example of a slice is given in Figure~\ref{fig-slice-fibo0}, for
the Fibonacci program $FIBO_0$ compiled with option $O0$: only $12$
instructions out of $40$ need be really simulated and the variables in
the sliced program are $\{pc,r_0,r_2,r_3\}$ and $3$ stack values.

\begin{figure}[thbtp]
  \centering
  \includegraphics[scale=0.45]{\fname{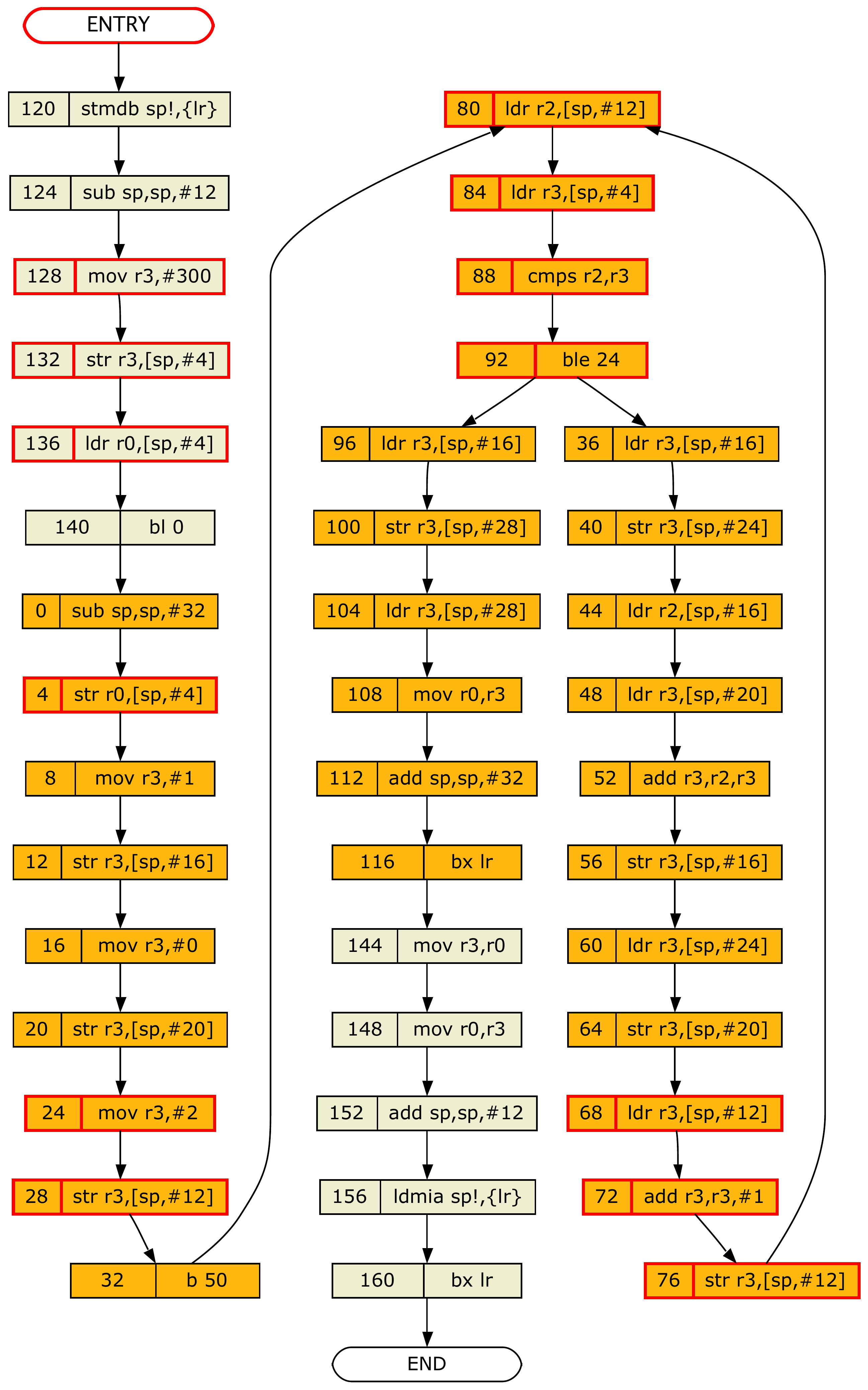}}
  \caption{WCET-equivalent Slice for Program $FIBO_0$.}
  \label{fig-slice-fibo0}
\end{figure}

Another advantage of slicing is that we do not need to do \emph{loop
  unrolling} because the registers and instructions that control the
loop bounds are automatically preserved by the slice.

In Table~\ref{tab-results}, Section~\ref{sec-experiments}, column
\emph{``Abs''} ($a/b$) gives, for each program $P$, the number of
nodes $a$ for which the \emph{simulation} of an instruction is needed
compared to the total number of nodes $b$ of $Aut(P)$.

This reduction has not only an effect on the state space (reduction of
the number of paths explored) but also on the size of the
representation of each state of $Aut(\alpha(P))$.

In the next section, we describe how we automatically compute the CFG
of a program.


\section{Computation of the CFG}
\label{sec-cfg}
To compute the CFG of a program, we iterate two phases:
\begin{enumerate}
\item Slice. We slice a partial CFG in order to compute the 
  dynamically computed branch targets; we simulate the sliced
  program to determine these targets.
\item Expand: having determined the dynamically computed branch
  targets, we expand the partial CFG and repeat Step~1.
\end{enumerate}
When the iteration terminates we have the CFG of the program.  We
limit the scope of our tool to non recursive programs, and this
ensures that the previous iterative computation terminates.

We describe the process on an example of a Fibonnaci program $FIBO_0$
(compiled with option $-O0$) given in Listing~\ref{lis-fibo0}.  This
program is composed of two functions, \emph{main} and \emph{fib}:
\emph{main} calls \emph{fib} and at the end, \emph{fib} returns.  The
computation would go like this: after instruction $8c$ in main,
\emph{fib} starts as $8c$ is ``(b)ranch to $0$ and save return address
to (l)ink register \emph{lr}''.  If at some point, instruction $74$ in
\emph{fib} is reached, \emph{lr} should contain the return address in
\emph{main} \ie $90$.  It should also be noticed that the first
instruction in \emph{main} is to save on the stack, the return of the
caller: \texttt{push(lr)}.  This is used at the end of \emph{main} to
return to the caller's next instruction when the statement \texttt{bx
  lr} (``branch to the content of \emph{lr}'') is performed right
after popping the value of \emph{lr}.

\begin{center}
  \begin{minipage}[htb]{0.55\linewidth}
\begin{lstlisting}[language=AssemblerARM9,numbers=none,caption={$FIBO_0$},label={lis-fibo0},basicstyle={\fontsize{8}{9}\selectfont\ttfamily}]
00000000 <fib>:
   0:   e24dd020        sub     sp, sp, #32
   4:   e58d0004        str     r0, [sp, #4]
   8:   e3a03001        mov     r3, #1
  12:   e58d3010        str     r3, [sp, #16]
  16:   e3a03000        mov     r3, #0
  20:   e58d3014        str     r3, [sp, #20]
  24:   e3a03002        mov     r3, #2             
  28:   e58d300c        str     r3, [sp, #12]  
  32:   ea00000a        b       50  <fib+0x50>
  36:   e59d3010        ldr     r3, [sp, #16] 
  40:   e58d3018        str     r3, [sp, #24] 
  44:   e59d2010        ldr     r2, [sp, #16]
  48:   e59d3014        ldr     r3, [sp, #20]
  52:   e0823003        add     r3, r2, r3
  56:   e58d3010        str     r3, [sp, #16]
  60:   e59d3018        ldr     r3, [sp, #24]
  64:   e58d3014        str     r3, [sp, #20]
  68:   e59d300c        ldr     r3, [sp, #12] 
  72:   e2833001        add     r3, r3, #1    
  76:   e58d300c        str     r3, [sp, #12] 
  80:   e59d200c        ldr     r2, [sp, #12] 
  84:   e59d3004        ldr     r3, [sp, #4]  
  88:   e1520003        cmp     r2, r3        
  92:   dafffff0        ble     24  <fib+0x24> 
  96:   e59d3010        ldr     r3, [sp, #16]
 100:   e58d301c        str     r3, [sp, #28]
 104:   e59d301c        ldr     r3, [sp, #28]
 108:   e1a00003        mov     r0, r3
 112:   e28dd020        add     sp, sp, #32
 116:   e12fff1e        bx      lr

00000078 <main>:
 120:   e52de004        push    {lr}   ; [stmdb sp!,{lr}]  
 124:   e24dd00c        sub     sp, sp, #12
 128:   e3a03f4b        mov     r3, #300      
 132:   e58d3004        str     r3, [sp, #4] 
 136:   e59d0004        ldr     r0, [sp, #4] 
 140:   ebffffdb        bl      0   <fib>
 144:   e1a03000        mov     r3, r0
 148:   e1a00003        mov     r0, r3
 152:   e28dd00c        add     sp, sp, #12
 156:   e49de004        pop     {lr}   ; [ldmia sp!,{lr}]  
 160:   e12fff1e        bx      lr

\end{lstlisting}
  \end{minipage}
\end{center}

If we perform a first \emph{unfolding} of the program, we obtain a
partial CFG depicted in Fig.~\ref{fig-fibo0-cfg1}.  In this CFG, the
successor of instruction $116$ is unknown and thus the unfolding has a
terminal node at this location.  To compute the successor of this
insruction we slice the partial CFG with the slice criterion
$\calC''=\{116\}$ and $\calV(116)=\{lr\}$.  The sliced program is
composed of the red nodes \ie instructions $140$ and $116$.
Simulating this two-instruction program we get the possible value of
\emph{lr} at instruction $116$ which is $144$.

We can then extend the partial CFG to obtain the graph depicted on
Fig.~\ref{fig-fibo0-cfg2}.  We slice again to compute the successor of
instruction $160$: the new slice (6 nodes) is depicted on
Fig.~\ref{fig-fibo0-cfg2} with the red nodes. We should here find that
\emph{main} handles the control back to its caller.  To recognise this
situation we use the following trick: we assume that before the first
instruction of the program is performed, $\sem{lr}=\beta$ where
$\beta$ is a special value that cannot correspond to any valid
instruction.  We can take for example $\beta=3$.  When we compute a
target which is $\beta$ we know that we have reached the end of the
program because this returns to the caller. This situation occurs when
we simulate the second slice and after the instruction \texttt{160:bx
  lr} the program return to the caller.

The complete CFG for $FIBO_0$ is given in Fig.~\ref{fig-fibo0-final}.

\begin{figure}[thbtp]
  \centering
  \includegraphics[scale=0.45]{\fname{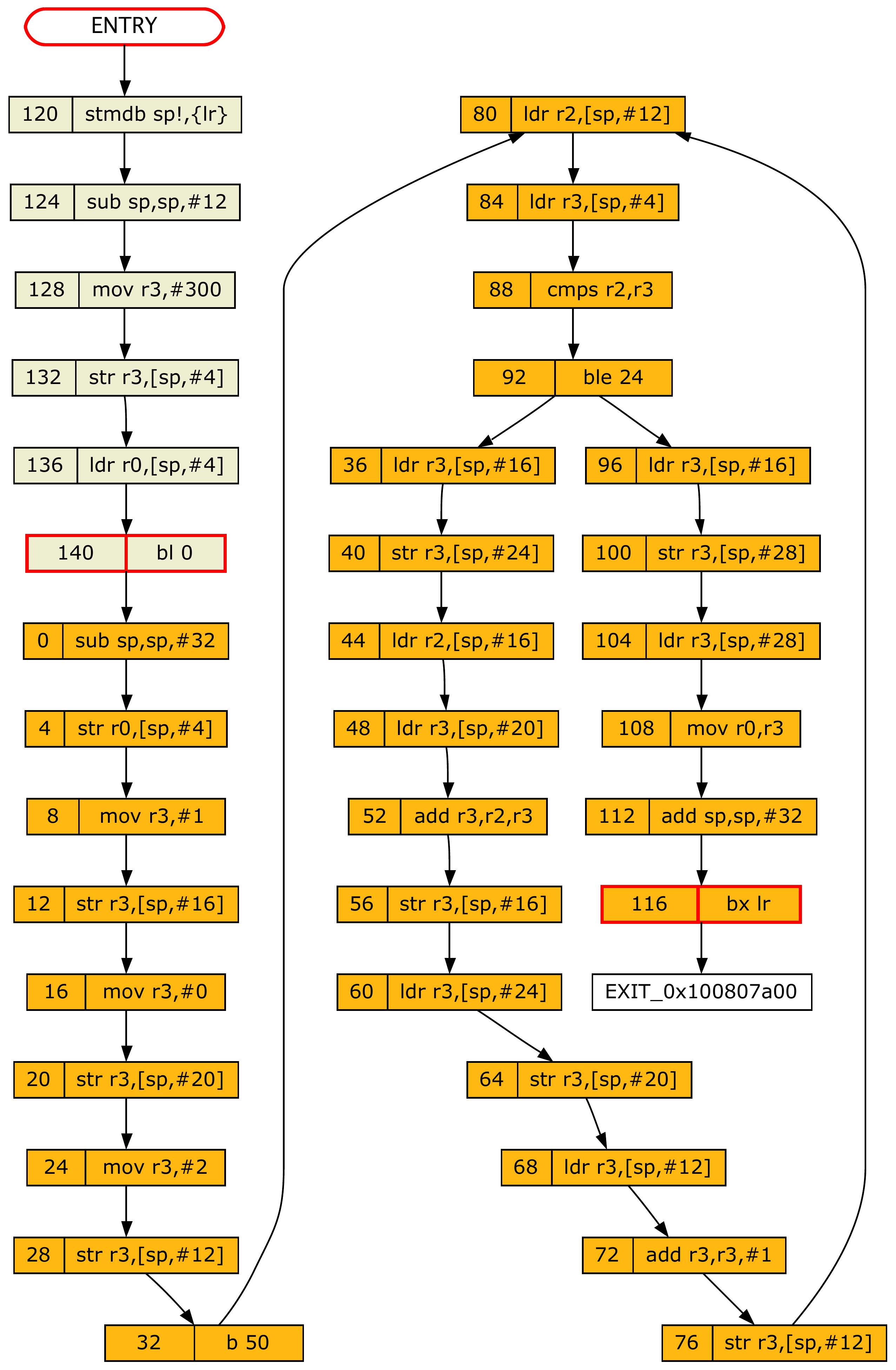}}
  \caption{First Unfolding of the CFG of $FIBO_0$.}
  \label{fig-fibo0-cfg1}
\end{figure}

\begin{figure}[thbtp]
  \centering
  \includegraphics[scale=0.45]{\fname{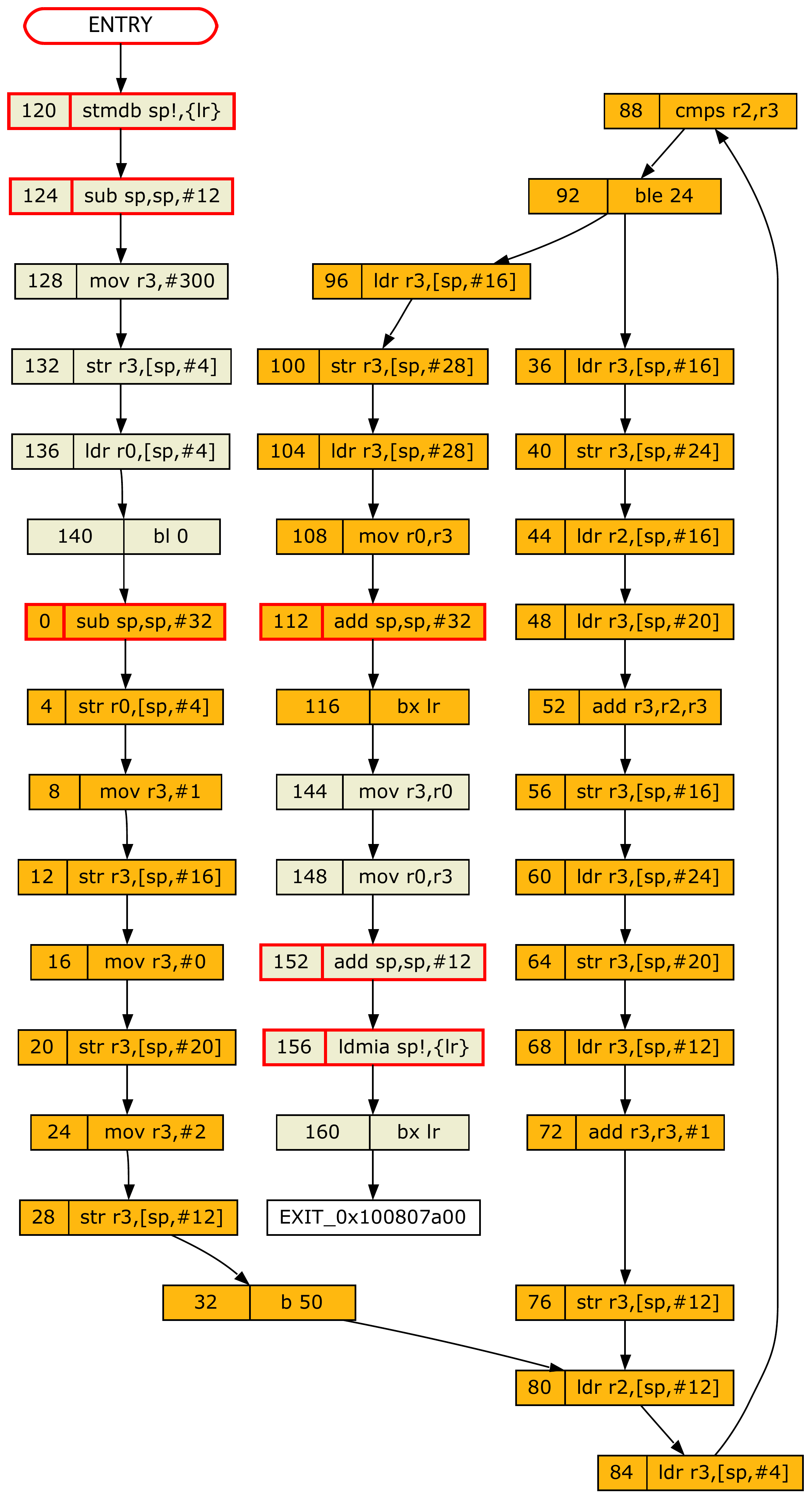}}
  \caption{Second Unfolding of the CFG of $FIBO_0$.}
  \label{fig-fibo0-cfg2}
\end{figure}

\begin{figure}[thbtp]
  \centering
  \includegraphics[scale=0.45]{\fname{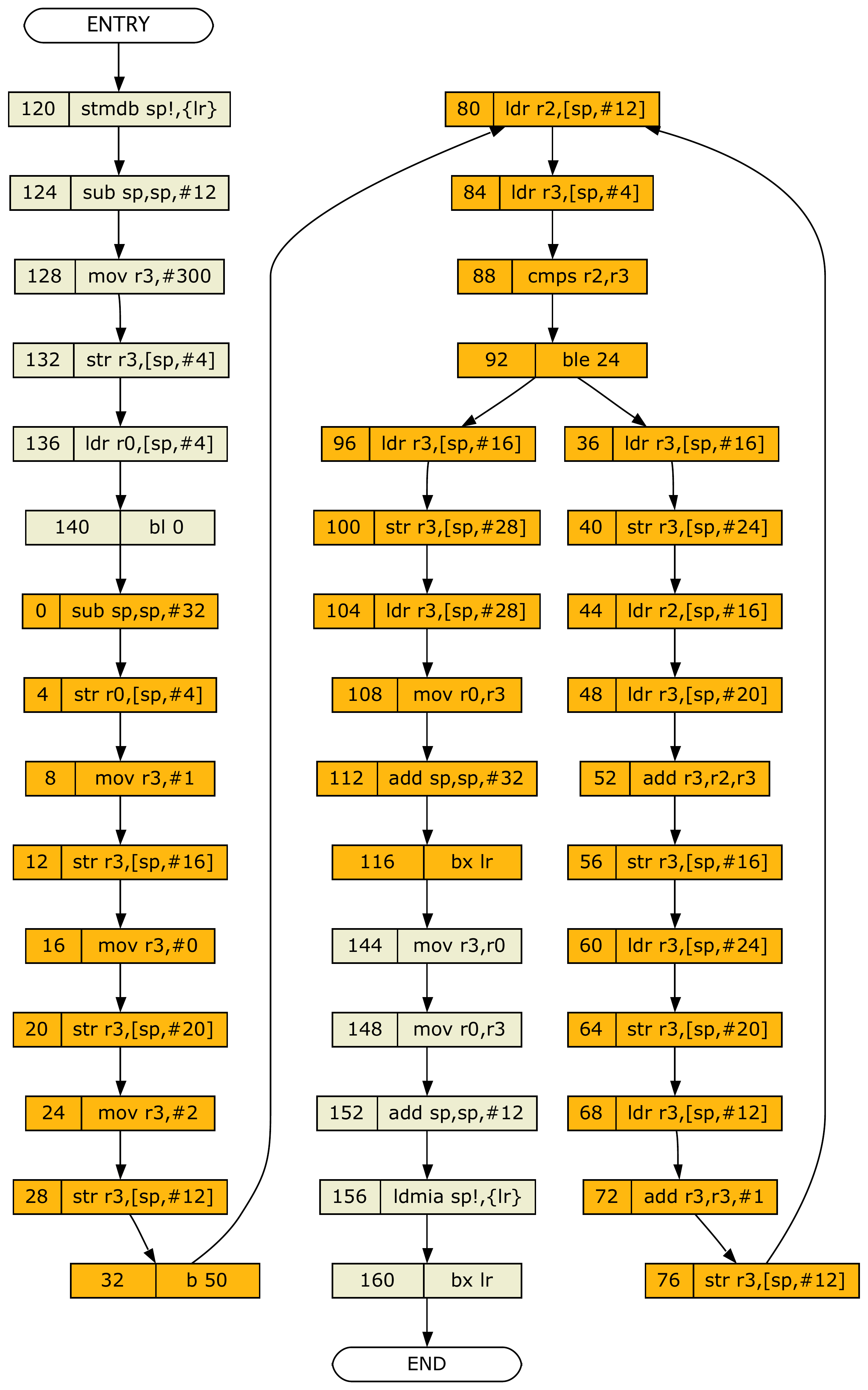}}
  \caption{The Complete CFG of $FIBO_0$.}
  \label{fig-fibo0-final}
\end{figure}

The computation of the possible values of \emph{sp} described in
Section~\ref{sec-step1} is actually performed when computing the
CFG. When we have computed the final CFG we also have the possible
values of register \emph{sp} (at the stack reference node) and we can
directly proceed to Step~2 (section~\ref{sec-wcet-equiv}) to compute a
WCET-equivalent program.

The previous process always converge to the CFG of a program because
we assume that the programs do not contain recursive calls (assumption
(A4)).  In the worst case, the slices we need to simulate in the
iterative compuation are the full CFGs obtained at each step.


\section{Hardware Model}\label{sec-hw}\label{sec-hw-model}


%
In this section we present some features of the formal models (timed
automata) of the hardware.  The automata are given using the \uppaal
syntax: \emph{initial} locations are identified by double circles,
\emph{guards} are green, synchronization \emph{signals} (channels) are
light blue and \emph{assignments} are dark blue. A \emph{C} in a
location means \emph{committed}: when an automaton enters a committed
location, it cannot be interrupted and proceeds immediately to one of
the successors of this location (the guards determine the transitions
that can be taken).  The \uppaal models are available from
\url{http://www.irccyn.fr/franck/wcet}.

\subsection{Main Memory}
\begin{figure}[thbtp]
  \centering
  \includegraphics[scale=1.2]{\fname{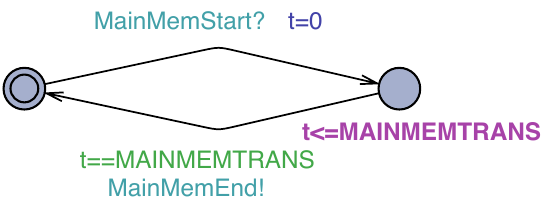}}
  \caption{Main Memory TA}
\label{fig-mainmem}
\end{figure}
The main memory model is a very simple two-location automaton
(Fig.~\ref{fig-mainmem}).  When a memory transfer is required, signal
\texttt{MainMemStart?} is received and clock $t$ is reset. After a
delay of \emph{MAINMEMTRANS} the transfer is completed and signal
\texttt{MainMemEnd!} is issued.  Main memory transfers are triggered
by either the instruction or data cache and accesses to main memory is
serialized.

\subsection{Caches}
The model of the instruction cache is given in Fig.~\ref{fig-icache}.
The state of the cache contains an array ($64\times 8$ array) to
record the addresses stored in the cache and whether a line is dirty
or not.

The instruction cache is simpler than the data cache because
no write can occur in this cache, so a line cannot be dirty.
After the initialization of the cache (initial state of the cache
by the function
\texttt{initCache()}), the automaton is ready for receiving the signal
\texttt{CacheReadStart[num]?}.  This signal will be triggered by the
fetch stage of the pipeline Fig.~\ref{fig-exec-decode-stage}.  The
memory address to read is $m$. If $m$ is in the cache (function
\texttt{is\_in(m)} returns \true), there is no need for a memory
transfer and variable \texttt{PMT} (Pending Memory Transfers) is
assigned $0$.  Otherwise function \texttt{insert(m)} inserts $m$ in
the cache and returns the number of memory transfers to be performed:
for the instruction cache it is always $1$ because a line cannot be
\emph{dirty} (see Section~\ref{sec-archi}) but for the data cache it
can be either one or $2$ if a dirty line has to be saved from the
cache.  As soon as the memory transfer is completed (\texttt{PMT}=0)
transition \texttt{Hurry!} is fired (it is \emph{urgent}). Then, after
\emph{CACHE\_SPEED} time units (value is $1$ for the our testbed) the
read request completes and the signal \texttt{CacheReadEnd[num]!}  is
issued.

\begin{figure}[hbtp]
  \centering
  \includegraphics[scale=1.3]{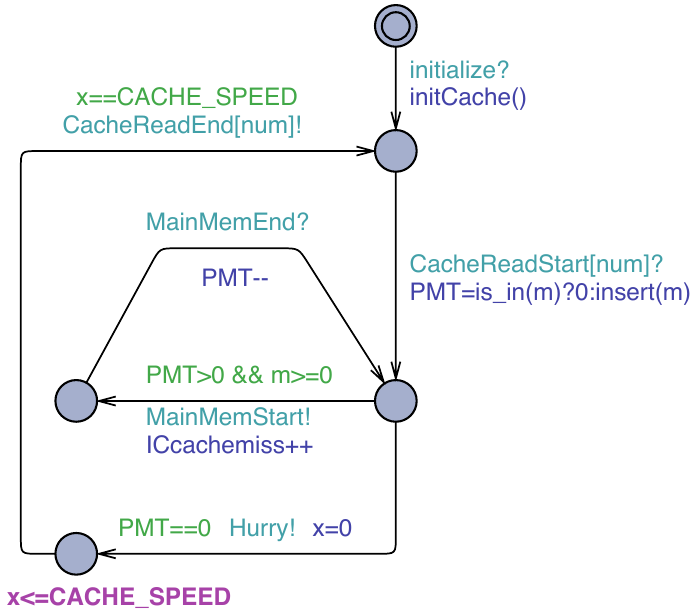}
  \caption{Instruction Cache}
\label{fig-icache}
\end{figure}

The data cache is a bit more involved (Fig.~\ref{fig-mem}).  For a
read/hit operation it behaves almost like the instruction cache
described above.  For write operations, a \emph{write buffer} (not
given here) is used and moreover the timing depends on the type
(load/store), addresses involved in the operation, and whether another
write/read operation is already in progress (and to which line in the
write buffer).  We have tried to design an accurate model of the data
cache: data cache operations are the major factor in the WCET for most
of the programs and a faithful model is required to compute tight
bounds.  How the model of the data cache was built is described in
Section~\ref{sec-tuning}.

\subsection{Pipeline Model}
The model of the pipeline is rather simple except the memory stage (M)
which is a bit more complicated.  The F stage automaton fetches the
next instruction if no branch delay stall occurs (see
Section~\ref{sec-tuning}).  The function \texttt{stall()} of the F stage
automaton determines whether such a stall should occur or not.  If the
next instruction can be fetched, it is fetched from the instruction
cache \texttt{CacheReadStart[INSTR\_CACHE]!} (this signal is urgent
and synchronized with the instruction cache).  When the fetch is
completed the instruction is transferred to the next D stage, as soon
as it is ready to be fed with a new instruction. The D stage, E stage
and W stage are similar.  Notice that the duration of an instruction
may vary from one instruction to the other (\eg long multiplication
may take longer than additions) or because a conditional instruction
is not executed: the actual duration is set when a new instruction
arrives in the E stage (\texttt{DUR\_INSTR=dur()}).  A special signal
\texttt{prog\_completed?} is received from the program and marks the
last instruction of the program. The program is completed when this
last instruction flows out of the last stage (W) of the pipeline and
this corresponds to reaching location \texttt{DONE} of the W stage.
\begin{figure*}[hbtp]
  \centering
\small
  \begin{tabular}[t]{c}
  \includegraphics[scale=1]{\fname{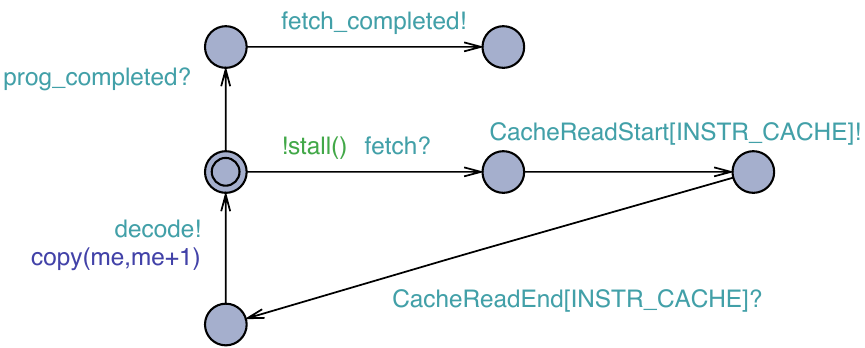}}\\
  F Stage
\end{tabular} \medskip
\begin{tabular}[t]{c}
  \includegraphics[scale=1.2]{\fname{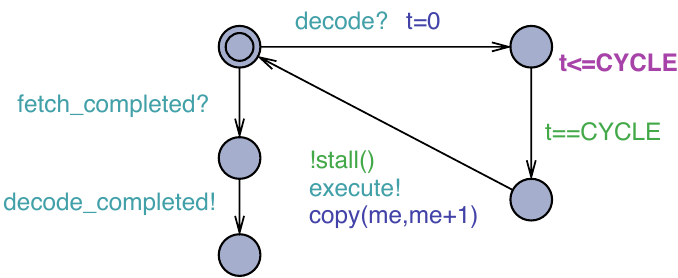}}\\
  D stage
\end{tabular} \medskip
\begin{tabular}[t]{c}
  \includegraphics[scale=1.2]{\fname{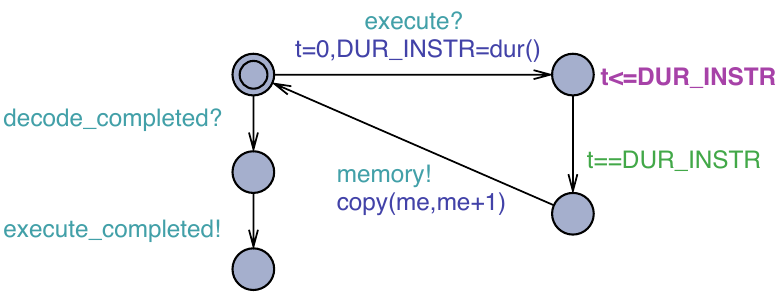}}\\
  E Stage
\end{tabular} \medskip
\begin{tabular}[t]{c}
  \includegraphics[scale=1.2]{\fname{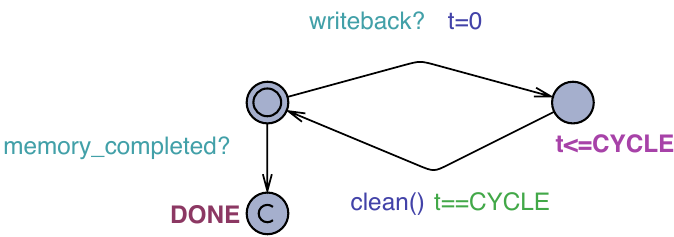}}\\
  W Stage
\end{tabular}
  \caption{Timed Automata for  the Fetch, Decode, Execute and WriteBack Stages.}
\label{fig-exec-decode-stage}
\end{figure*}
The automaton for the M stage is given in Fig.~\ref{fig-mem}: when an
instruction is performed and it is a memory transaction, it issues a
sequence of read/write requests to the data cache.

\begin{figure}[hbtp]
  \centering
  \includegraphics[width=0.9\linewidth]{\fname{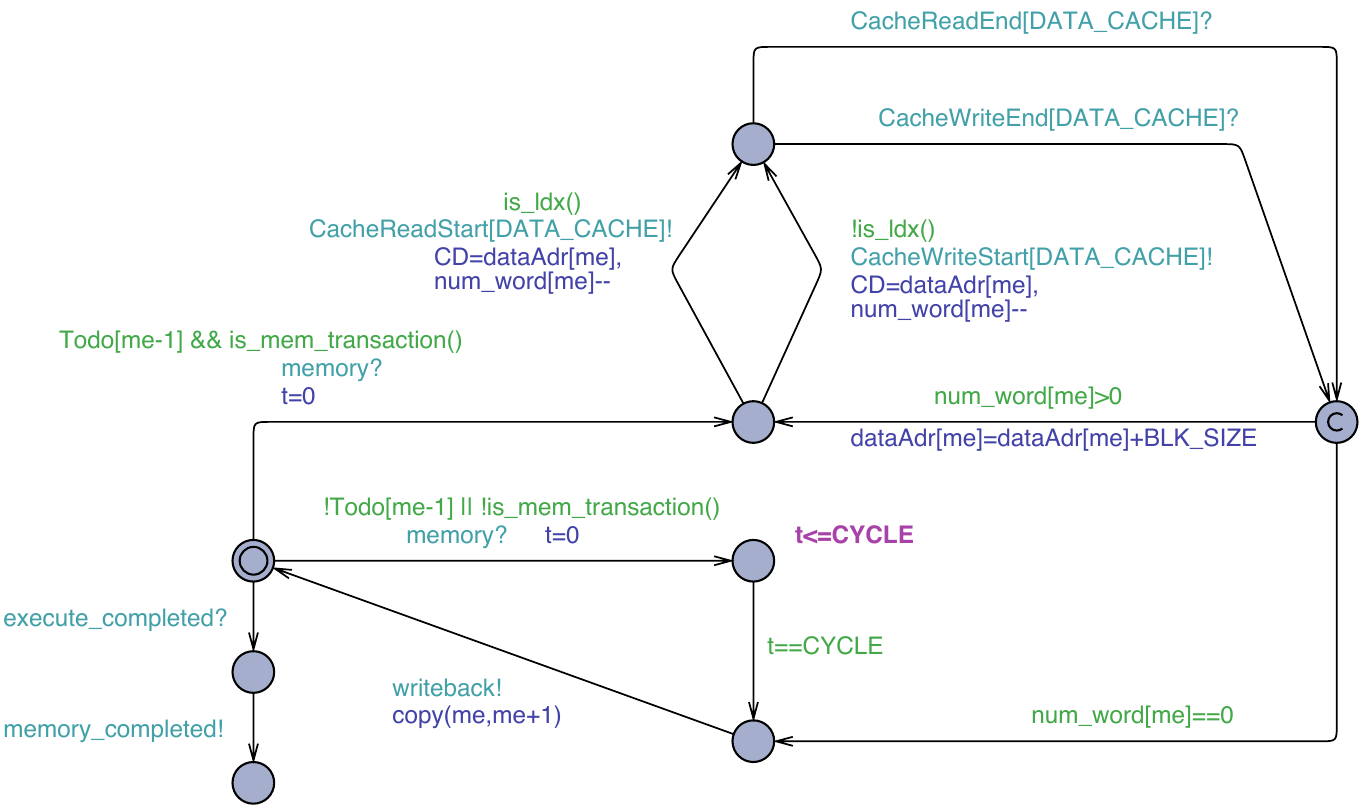}}
  \includegraphics[width=0.9\linewidth]{\fname{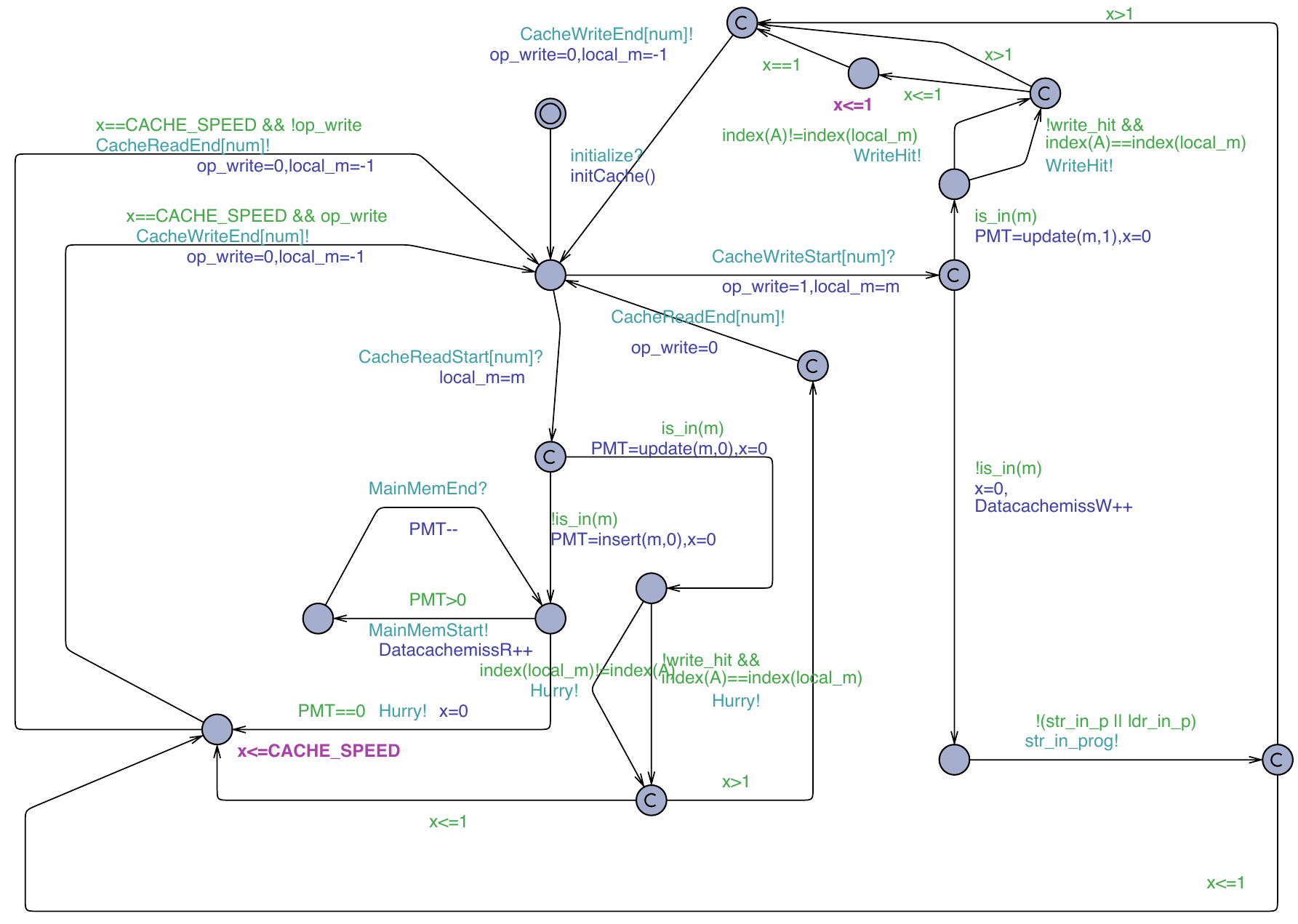}}
  \caption{Memory Stage and Data Cache}
\label{fig-mem}
\end{figure}


\section{Implementation}
\label{sec-implem}

We have implemented the construction of the CFG
(Section~\ref{sec-cfg}) and the computation of the WCET-equivalent
program (Section~\ref{sec-wcet}).
The architecture of our tool is given in Fig.~\ref{fig-tool-chain}. 
Together with a parser of ARM binary programs it comprises several
thousand C++ lines of code.  We have implemented very efficient
versions of \emph{post-dominators
  algorithms}~\cite{tarjan-79,Georgiadis-06} and \emph{post dominance
  frontiers algorithms}~\cite{cooper-01} as they are used intensively
both in \texttt{Compute CFG} and \texttt{Compute WCET-equiv}.  To
obtain the binary program we use the GCC tool suite (gcc, objdump)
from Codesourcery~\cite{codesourcery}.

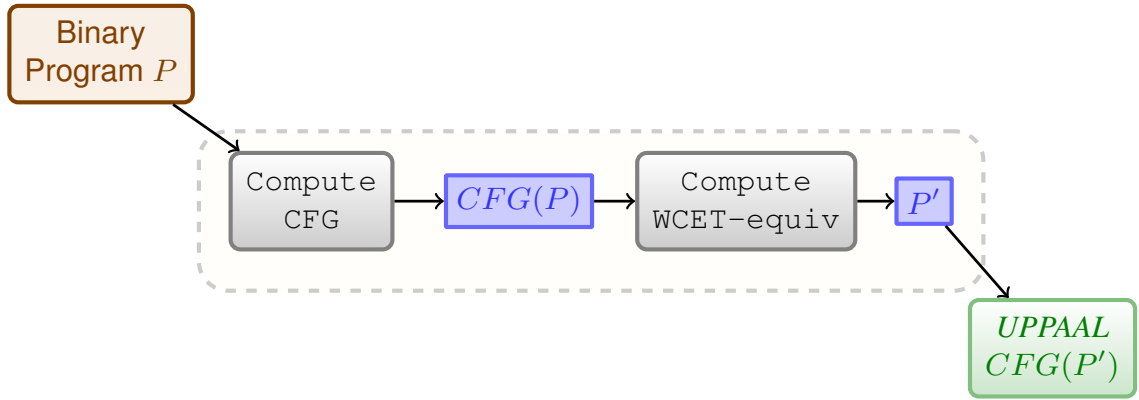
\begin{figure}[thbp]
\centering
\small
\scalebox{1.3}{\begin{tikzpicture}[
  input/.style={ 
    rectangle, 
    color=orange!50!black,
   draw=orange!50!black,
    rounded corners=0.9mm,
    minimum size=6mm, 
    very thick, 
    bottom color=orange!80!black!10, 
    top color=orange!70!black!10, 
    font=\sffamily
},
  output/.style={ 
    color=green!50!black,
    rectangle, 
    minimum size=6mm, 
    rounded corners=0.9mm,
    very thick, 
    draw=green!50!black!50, 
    top color=white, 
    bottom color=green!50!black!20, 
    font=\itshape 
},
inter/.style={color=blue!80,rectangle,rounded corners=0.1mm,very thick,fill=blue!20,draw=blue!60},
mycomp/.style={ 
  rectangle,minimum size=6mm,rounded corners=2mm, 
  very thick,draw=orange!50, top color=white,bottom color=orange!20,
  font=\ttfamily},
othercomp/.style={ 
  rectangle,minimum size=6mm,rounded corners=1mm, 
  very thick,draw=black!50, top color=white,bottom color=black!20,
  font=\ttfamily},
result/.style={ 
minimum size=6mm,rounded corners=0.1mm, 
  very thick,
  top color=white,bottom color=red!30,
  font=\sffamily},
 node distance=1cm and 1.9cm, thick]

\node (arm-file) [input] {
  \begin{tabular}[t]{c}
Binary \\ Program $P$
\end{tabular}
};
\node (compute-cfg) [othercomp,below right= of arm-file,yshift=-.5cm,xshift=.23cm] { \begin{tabular}[h]{c}
Compute \\ CFG 
\end{tabular}};
\node (cfg) [inter, right=of compute-cfg,xshift=.2cm] {$CFG(P)$};

\node (slice) [othercomp, right=of cfg,xshift=.4cm] {\begin{tabular}[h]{c}
Compute \\ WCET-equiv
\end{tabular}};
\node (pp) [inter, right=of slice,xshift=-.1cm] {$P'$};

\node (output) [output,below right=of pp,xshift=-.6cm,yshift=-.5cm] {
  \begin{tabular}[t]{c}
UPPAAL \\ $CFG(P')$
\end{tabular}
};

\path[->] (arm-file) edge[]  (compute-cfg) ;
\path[->] (compute-cfg) edge[swap]  (cfg) ;
\path[->] (cfg) edge[]  (slice) ;
\path[->] (slice) edge[]  (pp) ;
\path[->] (pp) edge[]  (output) ;

\begin{pgfonlayer}{background}
\node (hull) [fill=LemonChiffon!10,draw=black!20,dashed,rounded corners=3mm,dashed,inner sep=.3cm,very thick,rectangle,yshift=-.1cm,fit=(compute-cfg) (pp)] {}; 
\end{pgfonlayer}

\end{tikzpicture}
}
\caption{Tool Chain Overview}
\label{fig-tool-chain}
\end{figure}

Our tool produces a bundle of files: a ready-to-analyse file
containing the \uppaal timed automata models of the program $P'$ and
the hardware models\footnote{The layout of the CFG is produced using
  \emph{dot}, \url{http://www.graphviz.org/}.} $CFG(P')$; a \emph{dot}
file with the graph of $P'$ and a ready-to-compile C++ file that
contains a simulator of the program $P'$. This last file can be
compiled and used to compute useful information like the ranges of
registers.  Notice that during the first phase \texttt{Compute CFG} we
compute the range of the \emph{stack pointer} and thus the tool can
also be used as a stack analyser.  To compute the WCET we check
property R($K$) (Section~\ref{sec-reach}) using \uppaal.

For the binary programs we have analysed, the time it takes to compute
the output file from a binary program is negligible (less than a
second).  The automata of the programs of Table~\ref{tab-results} and
the \emph{dot} graphs are available from
\url{http://www.irccyn.fr/franck/wcet}.

\section{Experiments}\label{sec-experiments}

\subsection{Methodology}
The program $P$ to analyse is encapsulated in a template function: an
example of use is given for program FIBO in
Listing~\ref{lis:measurement}.
\begin{center}
\begin{minipage}[t]{0.56\linewidth}
\begin{lstlisting}[language=C,numbers=none,caption={Code snippet of instrumentation with FIBO},label={lis:measurement},basicstyle={\fontsize{7}{8}\selectfont\ttfamily}]
#define timerToCPUClockRatio 12

main ()
{
  int result;
  unsigned int start;
  unsigned int stop;
  
  start = timerGetValue(1);
  result = fib(300);
  stop = timerGetValue(1);
  printf("fib(300): %d, time=%lu\n", result, 
                (stop-start)*timerToCPUClockRatio);
  while (1);
}
\end{lstlisting}
\end{minipage}
\end{center}
Given $P$, we let $t(P)$ be the encapsulated program.  Measuring the
execution time of $P$ consists in (1) reading a \emph{hardware} timer
(\texttt{timerGetValue}) into a $start$ variable, (2) calling the
program $P$, and (3) reading the timer again into a $stop$ variable and
(4) printing\footnote{the Armadeus APF9328 board has a serial
  interface and in-rom drivers and printf function.} the difference
$stop - start$.  The function \texttt{timerGetValue} (assembly code)
has been designed to read a hardware timer (See next paragraph).

The measurement error is is \plusminus$12$ processor
cycles. 
The program $t(P)$ is compiled and linked.  Running it on the ARM9
will print out the number of cycles taken by the program $P$: this
figure is given in column ``Measured WCET'' in
Table~\ref{tab-results}.

To faithfully compute the WCET of $P$ using our method, we take as
input of our tool chain $t(P)$.  $t(P)$ is transformed (using Compute
CFG and Slice) into an \uppaal automaton as described in
Section~\ref{sec-implem}.  In this automaton a dedicated clock
\texttt{GBL\_CLK} is reset when the instruction\footnote{We can
  identify this instruction in \texttt{timerGetValue}.} of $t(P)$ that
reads the hardware timer flows out of the M stage (reading the timer
in function \texttt{timerGetValue} is done using a load instruction).
The final state of the automaton is reached when the second occurrence
of the instruction that reads the timer flows out of the W stage.  The
computed WCET is given in column ``Computed WCET'' in in
Table~\ref{tab-results}.  Column ``\uppaal'' in
Table~\ref{tab-results} gives the time \uppaal takes to check the
reachability property ``Is it possible to reach a final states with
\texttt{GBL\_CLK} $\geq K+1$ ?'' and this property is false and was
true for $K$. In this case $K$ is the computed WCET.

\subsection{Measuring Time on the Hardware}
\label{sec-measuring}

Measuring execution time on the hardware may be done by using an
external device like an oscilloscope or by using one of the embedded
hardware timers. In both cases, the program must be instrumented. In
the first case, using a General Purpose I/O (GPIO) device, a signal is
set to 1 at the start of the measure and to 0 at the end and the
oscilloscope measures the time between the rising and the falling
edge. In the second case, a free running timer is launched. It is read
at the start and at the end of the measure. The difference of both
values gives the execution time. This supposes the clock frequency of
the hardware timer is close enough to the clock frequency of the
processor to allow accurate measurements. By close enough we fix the
measurement error to less than \plusminus$1\%$ of the measurement. So
a hardware timer clock frequency two orders of magnitude lower than
the processor clock frequency would be accurate enough if the program
to measure executes in $\geq 10000$ cycles.

On the MC9328MXL the maximum available frequency for the hardware
timers is $\frac{1}{12}th$ the processor clock frequency. So a program
executing in $\geq 1200$ cycles may be accurately measured (less than
$1\%$ error).

\subsection{Tuning the Hardware Model}
\label{sec-tuning}
The ARM9TDMI Technical Reference Manual~\cite{arm9tdmi} gives
pipeline timings according to the kind of instructions together with
some examples of load delays and branch delays.
However these timing information about the ARM920T processor and the
MC9328MXL micro-controller are not enough to design accurate formal
models of the hardware.

To overcome this, we have carefully crafted programs to stress
particular features of the hardware and determine the precise timing
of some sequences of instructions. The basis of this
\emph{identification phase} consists in measuring the difference in
execution times of two variants of the same loop.  The second variant
contains a sequence of instructions for which we want a precise
timing. The execution time difference between the two variants is the
execution time of this sequence multiplied by the number of
iterations. Using a large number of iterations minimizes the
measurement error.

For memory accesses, variants may differ only by the memory alignment
of data because timings may be different if a subsequent cache access
is done in the same cache set or in a distinct cache set. And this can
have a huge impact on the computed WCET if not modelled properly.

To remove the execution time of the measurement code, the loop is
executed twice, one with 10000 turns and one with 20000 (for
instance). The difference of execution time is the execution time of
10000 turns. The loop is dried run to copy it into the instruction
cache.

Running a large set of special-purpose programs, we were able to refine
the model of the data cache and obtain a rather precise formal model
(see Fig.~\ref{fig-mem}).

\subsection{Test program example}

This methodology allowed us to work out an undocumented behavior of
the data cache. The loop in Listing~\ref{lis-dcache-timing} is executed
10000 times and 20000 times and the difference is 70000 cycles. This
result is consistent with the timing of the instructions found
in~\cite{arm9tdmi} since the instructions in the loop take 7 cycles to
execute (execution time of each instruction is given as comment in
listing \ref{lis-dcache-timing}).

\begin{center}
\begin{minipage}[t]{0.5\linewidth}
\begin{lstlisting}[language=AssemblerARM9,numbers=none,caption={Data cache timing behavior test},label={lis-dcache-timing},basicstyle={\fontsize{7}{8}\selectfont\ttfamily}]
.global ld_follow_st
ld_follow_st:
  ldr r2,[r1,#0]		%@ preload both addresses%
  ldr r2,[r1,#16]		%@ in the data cache%
ld_follow_st_loop:
  str r2,[r1,#0]			%@ 1 cycle%
  ldr r2,[r1,#16]			%@ 1 cycle%
  sub r0,r0,#1			%@ 1 cycle%
  cmp r0,#0			%@ 1 cycle%
  bgt ld_follow_st_loop		%@ 3 cycles%
  bx lr
\end{lstlisting}
\end{minipage}
\end{center}

However when the argument passed in {\em r1} (the base address used to
do the store and the load) is offset by 16 bytes, the execution time
is 80000 cycles because the instructions in the loop take 1 extra
cycle to execute.

The data cache has 64 sets and 32 bytes per line. So, the index is
located in bits 10 to 5 of the address. In the first case, with
$\sem{r_1} = 0x8004d94$ and $\sem{r_1}+16 = 0x8004da4$, the indexes are
different.
%
In the second case, with $\sem{r_1} = 0x8004da4$ and $\sem{r_1}+16 =
0x8004db4$, the indexes are equal.
%
So, after a store in a set, an access to the same set incurs a 1 cycle stall.

\subsection{Experiments on Benchmark Programs}
\label{sec-bench}
The results we have obtained on some benchmark
programs\footnote{\url{http://www.mrtc.mdh.se/projects/wcet/benchmarks.html}}
from \malar~\cite{Gustafsson:WCET2010:Benchmarks} are reported in
Table~\ref{tab-results}.  The programs we have analysed are available
from \url{http://www.irccyn.fr/franck/wcet}: we have archived the C
source program, the (de-assembled) encapsulated binary program (.arm
file), the \uppaal model (and property) and the \emph{dot} graph.  We
have not given the time it takes to do the slicing because it is less
than a second.  Regarding the benchmarks themselves, we point out
that:
\begin{itemize}
\item the difficulty of measuring the WCET is not related to the size
  of the program; some programs are huge but contain a few paths,
  others are very compact but have a huge number of paths.
\item they are designed to be representative of the difficulties
  encountered when computing WCET: for instance \textbf{janne-complex}
  contains two loops and the number of iterations of the inner loop
  depends on the current value of the counter of the outer loop (in a
  non regular way).
\item we have experimented on different compiled versions of the same
  program (options O0, O1, O2) because the binary code produced
  stresses different parts of the hardware.
\item we have checked various cases of the same programs with
  different initial stack pointer alignment, \ldots
\item we have multiplied the number of iterations of the benchmarks
  (\eg we compute the execution time of $Fib(300)$\footnote{Even if we
    cannot compute $Fib(300)$ we can compute the time it takes to
    compute it.}); this way a modelling error (\eg that adds $1$ cycle
  per iteration) is revealed and will incur a huge
  over-approximation.
\end{itemize}
In this sense the programs we have experimented on should not be
considered too easy.

The results in Table~\ref{tab-results} are divided into three main
sections:
\begin{itemize}
\item \emph{Single-Path programs.} The results of this section show
  that the abstract models (program and hardware) we have designed are
  adequate for obtaining tight bounds for the WCET. Even for
  \textbf{janne-complex} and its intriguing inner loop counts that
  depend on the outer loop counter, the maximum error is $3.2\%$.
  This also validates the accuracy of the program model we have
  computed (using slicing and no loop unrolling nor maximum loop
  bounds).
\item \emph{Single-Path programs with data dependent instruction
    durations.}  Instructions like MUL/MLA can take between $3$ to $6$
  cycles in the E stage (and SMULL $4$ to $7$). This section
  highlights one of the advantages of the timed automata models of the
  hardware. Indeed, in the timed automaton of the E stage
  (Fig.~\ref{fig-exec-decode-stage}), we can replace the guard
  \texttt{t==DURATION} with \texttt{MINDUR<= t <= MAXDUR} and (add the
  assignments to \texttt{MINDUR} and \texttt{MAXDUR}).  With this new
  E stage, we compute an \emph{interval} for the WCET.  Notice that
  this model is robust against \emph{timing anomalies} because we
  explore the state space without any assumption like ``always the
  shortest duration'' or ``always the largest duration''; the duration
  of the instruction is picked non-deterministically in
  [\texttt{MINDUR},\texttt{MAXDUR}] every time the transition is
  taken.  This explains the difference between the computed and the
  measured WCETs because in the measured WCET the worst-case duration
  for the MUL/MLA/SMULL instructions is never encountered.  In this
  case, column $\frac{(C-M)}{M}$ of Table~\ref{tab-results} does not
  represent the over-approximation of the computed WCET but rather the
  under-approximation of the measured WCET with the chosen input data.
\item \emph{Multiple-path programs.} These programs contain some
  branching that are input data dependent. The measured WCET is the
  execution time (on the hardware) obtained with input data that are
  supposed\footnote{Note that the benchmark programs usually indicate
    which data should give the WCET but in some cases this is
    erroneous.}  to produce the WCET.  The computed WCET result
  considers all the possible input data.  For \textbf{bs-O0,O1,O2} the
  WCET is very small and measurement errors are more than $1\%$ (see
  Section~\ref{sec-measuring}).  Program \textbf{cnt} starts with the
  initialization of a $10\times10$ matrix. In \textbf{cnt-O2}, the
  compiler unrolls the initialization loop to a list of 100
  consecutive store instructions. So \textbf{cnt-O2} stresses the
  write buffer and we have to take into account the fact that the
  Write Buffer may be full. In this case, the data cache has to wait
  to make a write until the write buffer is not full.

\end{itemize}
Compared to existing methods and results our method has several
advantages:
\begin{itemize}
\item computation of the CFG and of the reduced program automaton is
  fully automated (no loop bounds annotation needed);
\item we use concrete caches and a detailed models of the hardware;
\item the model of the hardware can be tuned easily (\eg durations of
  instructions can be an interval instead of a fixed value); as
  emphasised in~\cite{cassez-acsd-11}, changes in the processor speed
  can also be modelled easily (using a timed automaton that sets the
  processor speed). This enables us to compute WCET with \emph{power}
  related constraints.  Another advantage is that changing the
  processor (\eg ARM7) requires only to change the pipeline automata.
\item we compare the computed results to actual execution times using
  a rigorous protocol. The relative error in the computed results can
  be assessed and the results show that our method and models give very
  tight bounds.
\end{itemize}

\newcommand{\myahref}[1]{{~\textbf{#1}}}
\begin{table*}[thbtp]
  \centering
  \begin{tabular}{||l||c|c|c|c|c|c|c||}\hline\hline
    ~\textbf{Program}~~ &
    \begin{tabular}[h]{c}
loc$^\dagger$
\end{tabular}
& 
\begin{tabular}[h]{c}
UPPAAL \\ Time/States Explored$^\P$
\end{tabular} & 
\begin{tabular}[h]{c}
Computed \\ WCET (C) 
\end{tabular}
 &  
\begin{tabular}[h]{c}
  Measured \\ WCET (M) 
\end{tabular} &
\begin{tabular}[h]{c}
$\frac{(C-M)}{M} \times 100$
\end{tabular}
& ~~Abs$^\S$~~ \\\hline\hline

     \multicolumn{7}{||c||}{\textbf{Single-Path Programs}}  \\ \hline\hline
    ~\myahref{fib-O0} & 74  &  1.74s/74181 & 8098  & 8064 & \textbf{0.42\%} & 47/131\\ \hline
    ~\myahref{fib-O1} & 74  &  0.61s/22332 & 2597  & 2544 &  \textbf{2.0\%} & 18/72\\ \hline
    ~\myahref{fib-O2} & 74  & 0.3s/9710 & 1209  & 1164 &  \textbf{3.8\%} & 22/71\\ \hline
    ~\myahref{janne-complex-O0}$^\ast$~ & 65 &  1.15s/38014 & 4264 & 4164  &  \textbf{2.4\%} & 78/173\\ \hline
    ~\myahref{janne-complex-O1}$^\ast$~ & 65  & 0.48s/14600 & 1715 & 1680  &  \textbf{2.0\%} & 30/89\\ \hline
     ~\myahref{janne-complex-O2}$^\ast$~ & 65  &  0.46s/13004 & 1557 & 1536  &  \textbf{1.3\%} & 32/78\\ \hline
    ~\myahref{fdct-O1} & 238   &  1.67s/60418 & 4245  &  4092 &  \textbf{3.7\%}  & 100/363 \\ \hline
    ~\myahref{fdct-O2} & 238   &  3.24s/55285 &  19231 & 18984   &  \textbf{1.3\%} & 166/3543 \\ \hline\hline
    \multicolumn{7}{||c||}{\textbf{Single-Path Programs$^\ddagger$ with MUL/MLA/SMULL instructions (instructions durations depend on data)} } \\ \hline\hline
    ~\myahref{fdct-O0} & 238  &   2.41s/85007 &  [11242,11800]  & 11448   &  \textbf{3.0\%} & 253/831 \\ \hline\hline
    ~\myahref{matmult-O0}$^\ast$  & 162  & 5m9s/10531230  &  [502850,529250]  & [511584,528684]  &   \textbf{0.1\%} & 158/314 \\ \hline
    ~\myahref{matmult-O1}$^\ast$ & 162  & 1m32s/1122527  & [130001,156402]  & [127356,153000]  & \textbf{2.2\%} & 71/172 \\ \hline
    ~\myahref{matmult-O2}$^\ast$ & 162  &  43.78s/1780548  &  {[122046,148299] }  & [116844,140664]  &  \textbf{5.4\%} & 75/288 \\ \hline
    ~\myahref{jfdcint-O0}  & 374 &   2.79s/100784 &  {[12699,12699]}  & 12588 &  \textbf{0.8\%}  & 159/792  \\ \hline
    ~\myahref{jfdcint-O1}  & 374 &   1.02s/35518 &  {[4897,4899]}  & 4668 &  \textbf{7.0\%}  & 25/325  \\ \hline
    ~\myahref{jfdcint-O2}  & 374 &   5.38s/175661 &  [16746,16938]   & 16380 &  \textbf{3.4\%}  & 56/2512  \\ \hline
    \multicolumn{7}{||c||}{\textbf{Multiple-Path Programs}} \\ \hline\hline
    ~\myahref{bs-O0} & 174 &   42.6s/1421474 &  1068   & 1056  &  \textbf{1.1\%}  & 75/151 \\ \hline
    ~\myahref{bs-O1} & 174 &   28s/1214673 &  738   & 720  &  \textbf{2.5\%}  & 28/82 \\ \hline
    ~\myahref{bs-O2} & 174 &   15s/655870 &  628  & 600  &  \textbf{4.6\%}  & 28/65 \\ \hline
    ~\myahref{cnt-O0}$^\ast$ & 115 &  2.3s/76238 & 9028   & 8836 &  \textbf{2.1\%} & 99/235 \\ \hline
    ~\myahref{cnt-O1}$^\ast$ & 115 &  1s/27279 & 4123  & 3996 &  \textbf{3.1\%} & 42/129 \\ \hline
    ~\myahref{cnt-O2}$^\ast$ & 115 &  0.5s/11540 & 3065   & 2928 &  \textbf{4.6\%} & 39/263 \\ \hline
    ~\myahref{insertsort-O0}$^\ast$ & 91 &  10m35s/24250737  &  3133  & 3108 & \textbf{0.8\%}  & 79/175 \\ \hline
    ~\myahref{insertsort-O1}$^\ast$ & 91 &   7m2s/11455293  &  1533  & 1500 & \textbf{2.2\%}  & 40/115 \\ \hline
    ~\myahref{insertsort-O2}$^\ast$ & 91 &    11.5s/387292 &  1371  & 1344  &  \textbf{2.0\%}  & 43/108 \\ \hline
    ~\myahref{ns-O0}$^\ast$ & 497 &   83.4s/3064315 &  30968   & 30732 & \textbf{0.8\%} & 132/215 \\  \hline
    ~\myahref{ns-O1}$^\ast$ & 497 &  11.3s/368719  &    11701 & 11568 &  \textbf{1.1\%} & 61/124 \\  \hline
    ~\myahref{ns-O2}$^\ast$ & 497 &  29s/1030746 &  7343   & 7236 &  \textbf{1.4\%} & 566/863 \\  \hline\hline
  \end{tabular}
  \begin{flushleft}
    \rule{3cm}{.3pt}\\
    $^\dagger$lines of code in the C source file \\
    $^\ddagger$ $\frac{(C-M)}{M} \times 100$ computed using the upper bound for $C$ (see Section~\ref{sec-bench}). \\
    $^\S$Non Abstracted instructions/Instructions \\
    $^\ast$Program selected for the
    WCET Challenge 2006 \\
    $^\P$Time in min/seconds on Intel Dual Core i3 3.2Ghz 8GB RAM
  \end{flushleft}

\vspace*{-.3cm}
  \caption{Results. \lowercase{file-Ox indicates that file was compiled using gcc -Ox (optimization option).}}
  \label{tab-results}
\vspace*{-1cm}
\end{table*}




\section{Conclusion and Future Work}\label{sec-conclu}

In this paper we have presented a framework based on program slicing
and model-checking to compute WCET for programs running on
architectures featuring pipelining and caching.  We have exemplified
the method by providing formal models of the ARM920T.  Moreover we
have compared the computed results with actual execution times on the
real hardware.
Our method is
modular and altering the model of the hardware can
  be done easily using the timed automata models and the CFG
is computed automatically.

In some cases there are a huge number of paths to be explored and
there is no hope that an exhaustive search will compute any result in
a life-time. Examples of such programs are multiple-path programs (\eg
program binary sort) with a lot of input data dependent branchings.
To overcome this problem we are developing a \emph{branch and bound}
techniques.  
We are also currently extending the framework to handle:
\begin{itemize}
\item generation of \emph{traces}: \uppaal can generate a witness
  symbolic trace of a path yielding the WCET.  From this symbolic
  trace, we want to compute initial values of the input data that
  produce this trace.  This can be achieved using techniques similar
  to Counter Example Guided Abstraction Refinement
  (CEGAR)~\cite{clarke-cegar-acm-03}.
\item co-processor calls. This can be achieved by adding a timed
  automaton model of the co-processor.
\item for some programs like OS kernels, interrupts can be generated
  and trigger interrupt handlers.  Computing the WCET in this case is
  not easy as it requires a model of the interrupts arrivals \eg ``the
  interval between two interrupts of type $i$ is at least $t$ time units''.
  We can model interrupts arrivals using timed automata.
\end{itemize}


\paragraph{\bfseries Acknowledgements.} The authors wish to thank Tim Bourke for
the careful proof-reading of the paper and many helpful comments.

\bibliography{wcet}

\end{document}